\documentclass[journal,comsoc]{IEEEtran}
\pdfoutput=1

\makeatletter
\def\ps@headings{%
\def\@oddhead{\mbox{}\scriptsize\rightmark \hfil \thepage}%
\def\@evenhead{\scriptsize\thepage \hfil \leftmark\mbox{}}%
\def\@oddfoot{}%
\def\@evenfoot{}}
\makeatother
\usepackage[cmex10]{amsmath}
\usepackage{url}
\usepackage{cite}
\usepackage{amsmath,cases}
\usepackage{amssymb}
\usepackage{stackrel}

\usepackage{graphicx}
\usepackage[top=0.75in,bottom=1in,left=0.625in,right=0.625in]{geometry}
\usepackage{tikz}
\usetikzlibrary{%
  calc,
  arrows,%
  shapes.misc,
  shapes.arrows,%
  chains,%
  matrix,%
  positioning,
  scopes,%
  decorations.pathmorphing,
  shadows,%
  decorations.text
}
\usepackage{longtable}

\newtheorem{theorem}{Theorem}
\newtheorem{lemma}{Lemma}
\newtheorem{remark}{Remark}

\newenvironment{proof}{{Proof:}}{\hfill$\blacksquare$}

\begin{document}

\title{Adaptive Release Duration Modulation for Limited Molecule Production and Storage}

\author
{
Ladan Khaloopour, Mahtab Mirmohseni and Masoumeh Nasiri-Kenari\\
Department of Electrical Engineering, Sharif University of Technology, Tehran, Iran\\
Email: ladan.khaloopour@ee.sharif.edu, \{mirmohseni, mnasiri\}@sharif.edu
}

\maketitle

\begin{abstract}

The nature of molecular transmitter imposes some limitations on the molecule production process and its storage. As the molecules act the role of the information carriers, the limitations affect the transmission process and the system performance considerably. In this paper, we focus on the transmitter's limitations, in particular, the limited molecule production rate and the finite storage capacity. We consider a time-slotted communication where the transmitter opens its outlets and releases the stored molecules for a specific time duration to send bit ``1" and remains silent to send bit ``0". By changing the release duration, we propose an adaptive release duration modulation. The objective is to find the optimal transmission release duration to minimize the probability of error. We characterize the properties of the optimal release duration and use it to derive upper and lower bounds on the system performance. We see that the proposed modulation scheme improves the performance.
\end{abstract} 


\renewcommand*{\thefootnote}{}‎
\footnote{This paper was presented in part in the 2018 Iran Workshop on Communication and Information Theory (IWCIT) \cite{khaloopour2018adaptive}.}
\section{Introduction}
\IEEEPARstart{M}{olecular} communication (MC) is a promising communication technique among machines in nano-scale and has important applications in medicine, health-care and environmental sciences, where the conditions of transfer such as transmission range, nature of the medium, and possible transmitted carriers are consistent with molecules \cite{nakano2013molecular}.
MC has some similarities and differences with the classical communication.
The main difference is that in MC, molecules are carriers of information instead of electromagnetic or optical waves. Another important difference of MC with classical communication, which we focus on in this paper, is the transmitter's limitations.

The most important transmitter's limitation in MC is availability of molecules that must be released as information carriers. We study this problem for a common on/off keying (OOK) transmitter, where the transmitter
releases a specific concentration of molecules into the medium in order to send bit ``1" and remains silent in order to send bit ``0". However, the amount of available molecules at the transmitter is subject to some constraints due to the molecule production process. In fact, the molecules are produced at the transmitter with a limited-rate process, such as chemical reactions \cite{farsad2016comprehensive}. In addition, the amount of molecules that can be stored in the transmitter is limited and thus we face a finite storage capacity. The objective in this paper is to study the effect of these limitations on the performance of an MC system, considering the probability of error at the receiver.

 While molecular channels, receivers, and carriers are widely studied in the literature \cite{farsad2016comprehensive, malaka1995kinetic}, molecular transmitter's limitations are relatively unexplored. An ideal transmitter is considered in most of the works, which is a point source that can release any desired number of molecules in a very shot time at the beginning of a time-slot (\emph{i.e.}, an impulse function) \cite{noel2014improving, shahmohammadian2013nano}. There are few works that study the transmitter's challenges in MC \cite{nakano2013transmission, movahednasab2016adaptive, garralda2011simulation, jamali2016symbol, chude2015diffusion, arjmandi2016ion, bafghi2018diffusion}.
 The amplitude constraint on the number of transmitted molecules is studied in \cite{nakano2013transmission, movahednasab2016adaptive}. In \cite{garralda2011simulation}, the transmitter releases molecules in shape of a square or Gaussian pulse, where any number of molecules can be released during the pulse time (removing the impulse function constraint). The authors in \cite{jamali2016symbol} consider a point source transmitter with instantaneous molecule release but with random symbol interval. The focus of \cite{jamali2016symbol} is on synchronization schemes.
 
A more realistic transmitter model is introduced in \cite{arjmandi2016ion}, which is a point source with limited molecule production rate, located at the center of a spherical cell with finite storage capacity. There are many ion channels on its surface that can be opened or closed by applying an electrical voltage or a ligand concentration. The MC transmitter with these two limitations is also studied in \cite{bafghi2018diffusion} from information theoretic perspective and some bounds on its capacity are derived. Although the transmitter's limitations have been modeled in some of the above works, no specific transmission scheme has been proposed for these limitations.

In this paper, we focus on the transmitter's limitations, \emph{i.e.}, the limited production rate and finite storage capacity. We consider diffusion-based mediums as an important class of molecular transmissions where the molecules diffuse in the medium to reach the receiver. 
Our goal is to design a transmission scheme that has a good performance in terms of probability of error, considering the limitations. Our contribution and results are provided in two cases of the absence and the presence of inter-symbol interference (ISI):
\\
	\emph{{1. No-ISI case}}:
\\
$\bullet$ Transmitter design:
		we propose a coding scheme based on changing release durations (\emph{i.e.}, molecule release duration). This lets the transmitter to adapt the number of transmitted molecules to satisfy the transmitter constraints. 
\\
$\bullet$ Optimal release durations:
		we use the optimal maximum likelihood receiver, which has a threshold form with a fixed threshold, to formulate the optimal release durations: as the solution of an optimization problem which does not have a closed form solution. However, we characterize its important properties.
\\
$\bullet$ Optimal decoder:
		 we design the optimal decoder that minimizes the probability of error. Thereby, we obtain an adaptive threshold based receiver for the optimal release duration.
\\
$\bullet$ Probability of error analysis:
		we compute the probability of error. Using the properties of optimal release durations, we obtain upper and lower bounds on the probability of error.\\
	\emph{2. ISI case}:\\
$\bullet$  Transmitter design:
	similar to no-ISI case, we propose a coding scheme to adapt the number of transmitted molecules, satisfying the transmitter's constraints. 
\\	
$\bullet$  Optimal release durations:
	we formulate an optimization problem to find the optimal release durations for a fixed threshold receiver and obtain two sub-optimal solutions for the channel with one and two-symbol ISI.
\\	
$\bullet$  Optimal decoder:
	we design two sub-optimal decoders which use the sub-optimal release durations to minimize the probability of error for one and two-symbol ISI. Then, we derive the optimal adaptive threshold based receiver for the sub-optimal release duration.
\\
Furthermore, we provide numerical results:
		\\
$\bullet$  Storage capacity:
	we show that increasing the molecule storage capacity significantly improves the performance in terms of probability of error. 
\\	
$\bullet$  Receiver:
	it is shown that an adaptive threshold receiver performs better than a fixed threshold receiver as expected.
\\	
$\bullet$  Memory:
	While in our proposed adaptive strategy with variable release durations, the complexity of the system increases compared to non-adaptive strategy where release durations are fixed, we show that the required memories for the proposed transmitter to save optimal release durations and for the receiver to save the adaptive thresholds are finite.

The rest of this paper is organized as follows. Section \ref{model} presents a model for the transmitter's limitations. In Section \ref{problem_formulation}, a modulation scheme based on adaptive release duration is proposed. Main results in no-ISI case and ISI case are  provided in Sections \ref{sec_Main results no-ISI} and \ref{sec_main results ISI case}, respectively. Numerical results are presented in Section \ref{sec_numerical results}. Finally, in Section \ref{conclusion}, we conclude this paper.                                                       



\section{System model} \label{model}

We assume a time-slotted communication system with an OOK transmitter and the slot duration of $T$.
We consider a point-to-point molecular communication system including a transmitter, a fluid medium as the channel, and a receiver, shown in Fig. \ref{TxChRx}.
 \begin{figure}
	\centering
	\includegraphics[trim={1cm 6.2cm 3.7cm 4cm},clip, scale=0.25]{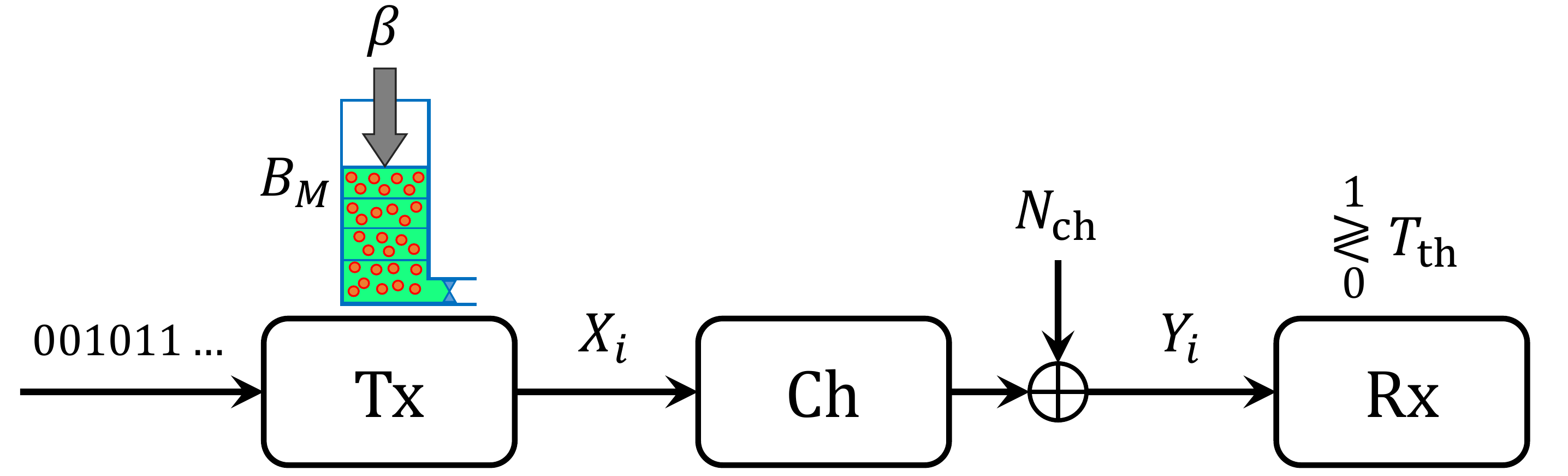}
	\caption{Molecular transmitter, channel and receiver}
	\label{TxChRx}
	\vspace{-1em}
\end{figure}
The transmitter releases a deterministic concentration of molecules into the channel. However, the transmitter has a limited molecule production rate $\beta$ and a finite molecule storage capacity $B_M$. The molecule storage is linearly recharged with a fixed rate $\beta$ at the transmitter up to its storage capacity $B_M$. We assume that the storage is fully charged in a duration less than slot duration ($B_M < \beta T$). Therefore, it is fully charged  if a ``0" is sent.

The released molecules diffuse in the channel according to Fick's second law of diffusion to reach the receiver, which results in a linear time-invariant channel.
We consider an absorbing receiver which absorbs molecules hitting its surface, with probability $p_k,\text{ } k = 0,1, \ldots$, in the current and next $k$ slots \cite{gohari2016information}. The reception process can be modeled as a Poisson process at the receiver \cite{mosayebi2014receivers}. In addition, we assume a Poisson background noise, $N_{\text{ch}}$, with parameter $\lambda$ as \cite{eckford2007nanoscale}. 
We consider a transmitter, located at $\vec{r}=0$, releases ${X}_i$ molecules at $i$th time-slot as a pulse train $\sum_{k}^{}{X}_i \delta(\vec{r}=0)\Pi(t-kT)$, where 
\begin{equation}
\Pi(t)=
\begin{cases}
1,\text{ } 0 \leq t \leq t_r, \\
0, \text{ }\text{otherwise},
\end{cases} \nonumber
\end{equation}
and $t_r$ is the release duration. Thus, the number of received molecules is
\begin{align}
Y_i \sim \text{Poisson}\left(\sum_{k=0}^{\infty} {p_{k} {X}_{i-k}}+\lambda\right). \label{poiss_dist}
\end{align}


\section{Problem formulation} \label{problem_formulation}

\subsection{No-ISI case:} \label{prob_formulate_no ISI}

First we study a simple case of no-ISI (\emph{i.e.}, $p_0=1, p_{k}=0 \text{ for }k\geq 1$).
We focus on a specific transmission scheme, where to send bit ``0", the surface outlets of the transmitter are closed and no molecule is released. To send bit ``1", the outlets are opened for some specific time duration and the stored molecules are released. Then the outlets are closed and the storage starts recharging.
In fact, as soon as the outlets are closed after the release duration, the transmitter starts storing the produced molecules.

\subsubsection{Non-adaptive communication strategy} \label{problem_formul_non-adapt}

First, consider a simple strategy where the release duration $T_M$ is fixed.
We set $T_M$ such that the storage is filled in time duration $T-T_M$
\begin{align}
T_M=\max_{B_M \geq \beta(T-t)} t =T-\frac{B_M}{\beta}.  \label{max_T1}
\end{align}
Thus, the storage is full at the beginning of each time-slot and $B_M=\beta(T-T_M)$.
In other words, $T_M$ is chosen such that $T-T_M$ is the minimum  transmitter required time to refill the storage; otherwise, we lose some molecule production time (because some molecules cannot be stored).
Assume that bit ``1" is to be transmitted at slot $i$.
As soon as the storage fills up to its capacity $B_M$, the outlets are opened for duration $T_M$, and $\beta T_M$ molecules are also produced at this time. Therefore the total number of released molecules for bit ``1" is 
$
{X_i}=B_M+\beta T_M=\beta T. \nonumber
$
If we quantize the release duration into very short intervals, substituting the above amount of $X_i$ in (\ref{poiss_dist}) and using thinning property of Poisson distribution result in
\begin{align}
Y_i \sim \text{Poisson}\left({\beta T}+\lambda\right). \label{Yi_Poiss_fixed}
\end{align}
At the receiver, we use the optimal maximum likelihood (ML) decoding approach. From (\ref{Yi_Poiss_fixed}), the ML receiver has a fixed threshold form. If we define $ M=\beta T$, the optimum threshold is derived as
\begin{align}
T_{\text{th}}=\frac{ M}{\text{ln}(1+\frac{ M}{\lambda})}. \nonumber
\end{align}
Since $ M \gg \lambda$, we have $T_{\text{th}}< M$. The receiver observes $Y_i=y$ and decides ${\hat X}_i=1$ if $y \geq T_{\text{th}}$ and  ${\hat X}_i=0$ if $y<T_{\text{th}}$.
Thus, the total probability of error is
\begin{align}
&P_e = \frac{1}{2}\left(P_{e|0}(0)+P_{e|1}(M)\right), \label{Pe_fixed} 
\end{align}
where 
\begin{align}
&P_{e|1}(x)=1-P_{e|0}(x)=\sum_{y< T_\text{th}}^{}e^{-(x+\lambda)}\frac{(x+\lambda)^y}{y!}. \nonumber
\end{align}

\subsubsection{Adaptive communication strategy} \label{problem_formul_adapt}

Now, we propose an adaptive strategy. We improve the system performance in terms of probability of error at the receiver by making the release duration adaptively variable. Thus, the storage may not be fully charged at the beginning of each time-slot. 
The duration of a time-slot is chosen such that transmitting a ``0" fully charges the storage. As a result, the transmitter has to know the number of subsequent ``1"s that has been transmitted prior to the current transmission. Thus, the system is state-dependent. We say that the system is in state $s_j$ if $j$ subsequent ``1"s are transmitted after the last transmitted ``0". In state $s_j$, we change the release duration $T_M$ to $T_M+\tau_{j+1}$. We call $\tau_i \in \mathbb{R}$ as release duration increments for $i=1, 2,\ldots$. Note that it can be negative to decrease the release duration. The release process in each slot begins when the storage is full (at $kT+\tau_1+\cdots+\tau_j$, in state $s_j$ for $k$-th transmitted bit).
The transmitter state machine is shown in Fig. \ref{TX state machine}.
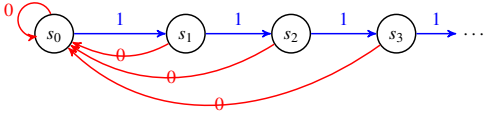
\begin{figure}
	\begin{center}
		\vspace{-0.5em}
		\scalebox{0.7}{%
			\begin{tikzpicture} [
			point/.style={coordinate},>=stealth',thick,draw=black,
			tip/.style={->,shorten >=0.007pt},
			text height=1.5ex,text depth=.25ex 
			]
			\path (2,8) node[draw,shape=circle] (A) {$s_0$};  
			\path (4.5,8) node[draw,shape=circle] (B) {$s_1$};
			\path (6.5,8) node[draw,shape=circle] (C) {$s_2$};
			\path (8.5,8) node[draw,shape=circle] (D) {$s_3$};  
			\path (10,8)	node (V) {$\cdots$}; 
			\draw [blue, ->] (A) -- node [above, midway] {$1$} (B);
			\draw [blue, ->] (B) -- node [above, midway] {$1$} (C);
			\draw [blue, ->] (C) -- node [above, midway] {$1$} (D);
			\draw [blue, ->] (D) -- node [above, midway] {$1$} (V);
			\draw [red, ->, bend left] (B) [out=30 , in=155] to node {$0$} (A);
			\draw [red, ->, bend left] (C) [out=32 , in=145] to node {$0$} (A);
			\draw [red, ->, bend left] (D) [out=35 , in=135] to node {$0$} (A);
			\draw [red, ->] (A.100) arc (20:277:3.2mm) node[pos=0.5,left] {$0$} (A);
			\end{tikzpicture}	}
	\vspace{-2em}
	\end{center}
	\caption{Transmitter finite state machine} 
	\label{TX state machine}
\vspace{-1em}
\end{figure}
The probability of state $s_j$ is
\begin{align}
P(s_j)=\frac{1}{2^{j+1}}. \label{state_probes}
\end{align}
Let $\Delta_n= \beta \tau_n, n\in \mathbb{N}$ denote the amount of the released molecules increment in state $s_{n-1}$. Therefore, the probability of error is
\begin{align}
P_e &=\frac{1}{2}\Big(P_{e|0}(0)+\sum_{j=1}^{\infty} P(s_{j-1})P_{e|1,s_{j-1} (M+\Delta_{j})}\Big)\nonumber\\
&\overset{(a)} =\frac{1}{2}\Big(P_{e|0}(0)+\sum_{j=1}^{\infty}\frac{1}{2^{j}}P_{e|1}(M+\Delta_{j})\Big),
\label{Pe0_Pe1_dTi}
\end{align}
where (a) follows from (\ref{state_probes}) and $P_{e|1,s_{j-1}}$ is the error probability of bit ``1" in state $s_{j-1}$.

Our goal is to design $\tau_i$s in order to minimize the probability of error. Because $P_{e|0}$ does not depend on $\tau_i$, we only consider $P_{e|1}$ and solve an optimization problem. Let us define 
\begin{align}
F(\{\Delta_i\}_{i=1}^{\infty}) \overset{\vartriangle}={\sum_{j=1}^{\infty}\frac{1}{2^{j}}P_{e|1}\left(M+\Delta_{j}\right)} \label{def_F(delta_i)},
\end{align}
then our problem is
\begin{align}
\min_{\{\Delta_i\}_{i=1}^{\infty}}&F(\{\Delta_i\}_{i=1}^{\infty}), \label{minF}\\
{\text{s. t. : }}&
	C_{1i}(\{\Delta_i\}_{i=1}^{\infty}): B_M-\sum_{j=1}^{i}\Delta_j\geq 0 ,\text{ }\forall i ,  \label{C_{1i}_no isi} \\  
	&C_{2i}(\{\Delta_i\}_{i=1}^{\infty}): \sum_{j=1}^{i}\Delta_j\geq 0 ,\text{ }\forall i.    \label{C_{2i}_no isi}                   
\end{align}
(\ref{C_{1i}_no isi}) indicates that the total increment is upper bounded by the total storage capacity ($B_M=\beta T$) and (\ref{C_{2i}_no isi}) is to ensure that we do not waste molecule production time. Due to limited storage capacity, if in state $s_i$, $\sum_{j=1}^{i}\Delta_i < 0$, we have $\sum_{j=1}^{i}\tau_j< 0$ and thus
\begin{align}
T- (T_M+\sum_{j=1}^{i}\tau_j)>\frac{B_M}{\beta}.   \nonumber             
\end{align}
It means that the storage is filled before the next slot and the production process stopped for a time duration.

Now, we motivate our adaptive communication strategy by showing that (in a simple example), increasing the release duration of non-adaptive communication strategy indeed improves the performance. Assume that bit ``1" must be transmitted in the $i$-th slot.
If $T_M$ is chosen as (\ref{max_T1}), the storage is full at the beginning of $(i+1)$-th time-slot but no molecule is needed if the transmitted bit at the $(i+1)$-th time-slot is ``0". It is possible to use this charged storage to increase the released molecules in the $i$-th time-slot (for bit ``1"). From (\ref{Pe0_Pe1_dTi}), it is observed that $P_{e|0}$ is fixed. However, $P_{e|1}$ is a decreasing function of the number of released molecules, because
\begin{align}
\frac{d}{dx}P_{e|1}(x)&=-e^{-(x+\lambda)}\frac{(x+\lambda)^{T_\text{th}}}{T_\text{th}!}<0, \label{d/dx P_e1}
\end{align}
where $x$ is the number of released molecules. As a result, if we increase release duration for sending ``1", then the number of released molecules increases and probability of error decreases.

As seen in Fig. \ref{DT_is_coloured},
\begin{figure}
	\centering
	\includegraphics[trim={1.7cm 7.1cm 1.3cm 7.6cm},clip, scale=0.29]{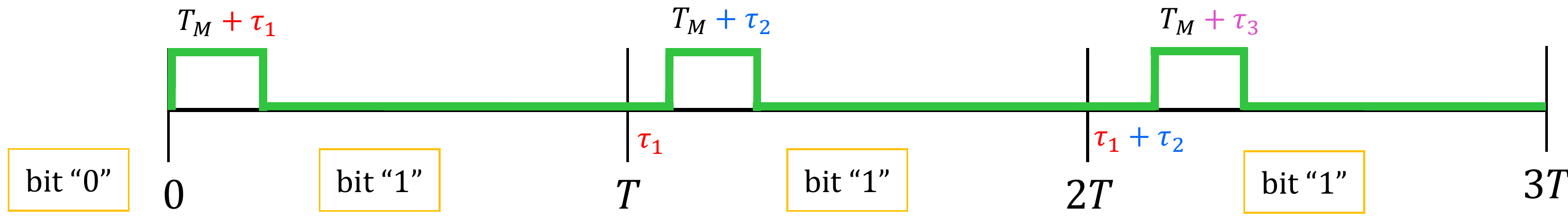}
	\caption{release duration increments}
	\label{DT_is_coloured}
	\vspace{-1em}
\end{figure}
at state $s_0$, if bit ``1" comes, we increase the release duration by $\tau_1$. So for the next bit ``1", the storage is not full at the beginning of time-slot. To compensate this reduction, we wait for time $\tau_1$ in order to the storage be filled. Then, the outlets are opened for release duration $T_M+\tau_2$.  Note that we should have $\tau_1+\tau_2 \geq 0$ not to waste molecule production time. For state $s_i$, we wait for time $\tau_1+\cdots+\tau_i$ in order to the storage be filled and then open the outlets for duration $T_M+\tau_{i+1}$. Totally, we increase the transmitted molecule number in state $s_i$ by $\beta \tau_i$ (which was $\beta T$ in the non-adaptive strategy). If $\tau_i$s are positive and $\tau_1+\cdots+\tau_{i-1} \leq T-T_M \text{ for }i=1,\cdots, J$, the number of released molecules increases in states $s_i  \text{ for }i=1,\cdots, J$ and remains fixed ($\emph{i. e.}, T_M$) for $i>J$. As a result, this scheme causes a smaller probability of error.

\subsection{ISI case:} \label{prob_formulate_ISI}

Similar to no-ISI case, our goal is to minimize $P_e$ subject to the constraints in (\ref{C_{1i}_no isi}) and (\ref{C_{2i}_no isi}).

If the channel memory is $K$ slots (\emph{i.e.}, $p_k=0, k>K$ in (\ref{poiss_dist})), and if we denote the increment of molecule number to send ``1" in state $s_i$ by $\Delta_{i+1}$, $P_e$ for a fixed threshold receiver can be written as the following form.
\begin{align}
P_e&=\frac{1}{2}\sum_{j=0}^{\infty}\text{P}(s_j)(P_{e|0,s_j}+P_{e|1,s_j})\nonumber\\ 
&=\frac{1}{2}\sum_{j=0}^{\infty}\frac{1}{2^{j+1}}\Big(P_{e|0}(w_{0}(\{\Delta_i\}_{i=1}^{K}))+P_{e|1}(w_{1}(\{\Delta_i\}_{i=1}^{K+1})\Big), \nonumber
\end{align}
where $w_{0}(\cdot)$ and $w_{1}(\cdot)$ are the number of received molecules when the current bit is ``0" and ``1", respectively.
\begin{align}
&w_{0}=\lambda+c_0M+c_1\Delta_1+\cdots+c_{K}\Delta_{K}, \nonumber\\
&w_{1}=\lambda+d_0M+d_1\Delta_1+\cdots+d_{K+1}\Delta_{K+1}.\label{weigthed_sum}
\end{align}
The $c_i$ and $ d_i$s are weighted sum of $M$ and $p_i$s, depending on the current and $K$ previous states and bits. We explain this weights through an example.
Assume that $K=4$ and bit sequence $(B_1,\ldots,B_{10})=(011101101)$ has been transmitted. Now, we send $B_{11}\in \{0,1\}$ and we have
\begin{align}
w_{0}&=(p_1+p_3+p_4)M+p_1\Delta_1+p_2\times 0 +p_3\Delta_2 +p_4\Delta_1 \nonumber\\
&=\underbrace{(p_1+p_3+p_4)}_{c_0}M+\underbrace{(p_1+p_4)}_{c_1}\Delta_1+\underbrace{p_3}_{c_2}\Delta_2.  \nonumber \\
w_{1}&=(p_0+p_1+p_3+p_4)M \nonumber\\
& \text{ }\text{ }\text{ }+p_0\Delta_2+p_1\Delta_1+p_2\times 0 +p_3\Delta_2 +p_4\Delta_1 \nonumber\\
&=\underbrace{(p_0+p_1+p_3+p_4)}_{d_0}M+\underbrace{(p_1+p_4)}_{d_1}\Delta_1+\underbrace{(p_0+p_3)}_{d_2}\Delta_2.  \nonumber
\end{align}



\section{Main results in no-ISI case} \label{sec_Main results no-ISI}
\subsubsection{Optimal solution}

It can be easily shown that the objective function $F(.)$ in (\ref{def_F(delta_i)}), is convex (see Appendix \ref{Convex objective function}) and the domain of optimization problem in (\ref{minF}) is compact ($\Delta_i \in [0, \beta(T-T_M)]\text{ for }i=1, 2, \cdots$). Thus, the global minimum exists. We start to find the solution using the Lagrangian method. We have the KKT conditions at the optimal point  (shown as $\{\Delta^*_i\}_{i=1}^{\infty}$), that the regularity condition is Linear Independence Constraint Qualification (LICQ). As a result, for the active conditions, we have
\begin{align}
\nabla F(\{\Delta_i\}_{i=1}^{\infty})=\sum_{i=1}^{\infty}\mu^*_{1i}\nabla C_{1i}+\sum_{i=1}^{\infty}\mu^*_{2i}\nabla C_{2i},  \label{g1}
\end{align}
where $\nabla$ denotes the vector differential operator ($\nabla=(\frac{d}{d\Delta_1},\frac{d}{d\Delta_2},\cdots) $), $\mu^*_{1i} \text{ and } \mu^*_{2i}$ are non-negative Lagrange multipliers for $i=1,2,\cdots$, and $C_{ki}, k=1,2$ have been defined in (\ref{C_{1i}_no isi}) and (\ref{C_{2i}_no isi}).
The above equation results in the following three properties for the optimal solution, stated in Lemmas 1, 2 and 3. The proofs of these lemmas are provided in Appendixes \ref{Appx_Proof_lem_property1}, \ref{Appx_Proof_lem_property2} and \ref{Appx_Proof_lem_property3}, respectively.

\begin{lemma} \label{lem_1}
The sequence of $\{\Delta^*_i\}_{i=1}^\infty$ is decreasing and nonnegative.
\end{lemma}
\begin{lemma}\label{lem_2}
There exists an index $J$ such that $\Delta^*_j=0$, for $j>J$. This means that a finite number of release duration increments are positive (\emph{i.e.}, non-zero).
\end{lemma}
\begin{remark}
The above \emph{lemma} concludes that we only need to change the release durations for a finite number of states. Hence, the transmitter needs a limited memory for saving the optimal values of these release duration increments. An upper bound on the number of release duration increments is derived in Section~\ref{Bounds for term numbers}.
\end{remark}

Note that from (\ref{d/dx P_e1}), $P_{e|1}$ is a decreasing function. Now, let $M+a_J$ be the point where $\frac{d}{dx}P_{e|1}(.)$ reaches half of its value at $x=M$ and similarly $M+a_i \text{ for }i=1,2,\cdots,J+1$ be the point where
\begin{align}
\frac{d}{dx}P_{e|1}(x)\Big|_{M+a_i}=\frac{1}{2^{J-i+1}}\frac{d}{dx}P_{e|1}(x)\Big |_{M}. \label{min_ai s}
\end{align}
These points determine the boundaries of the optimal solution in the following \emph{lemma}.
\begin{lemma} \label{lem_3}
	The optimal $\Delta^*_i $ belongs to the interval $ [a_{i+1},a_i]$ for $i=1,2,\cdots$. 
\end{lemma}
\begin{remark}
 	The above property helps us to provide upper and lower bounds on the minimum probability of error (see Section~\ref{Upper_Lower bounds for Pe}).
\end{remark}
\begin{remark} \label{remark:min Pe_ opt delta & fixed th}
	Substituting the optimal $\Delta^*_i$s in (\ref{Pe0_Pe1_dTi}) gives us the minimum probability of error for a fixed threshold receiver. 
\end{remark}


\subsection{Performance bounds}

In this section, we first derive the upper and lower bounds on the probability of error, and then we characterize an upper bound on the number of release duration increments.

\subsubsection{Upper and lower bounds on the minimum probability of error} 
\label{Upper_Lower bounds for Pe}

We provide upper and lower bounds on the probability of error using \emph{Lemma} \ref{lem_3}. Let $J$ be the number of positive release duration increments. Recall that each optimal $\Delta^*_i \in [a_{i+1},a_i]$. Using this and (\ref{positive_sum}), we obtain
\begin{align}
\sum_{j=2}^{J+1}a_j \leq \beta (T-T_M) \leq \sum_{j=1}^{J}a_j. \nonumber
\end{align}
First, we derive an upper bound.	
From (\ref{Pe0_Pe1_dTi}), we know that $P_{e|0}$ does not depend on $\Delta_i$, and from (\ref{d/dx P_e1}) we note that $P_{e|1}$ is a decreasing function of the number of transmitted molecules. Based on \emph{Lemma}~\ref{lem_3}, we have $\Delta^*_i>a_{i+1}$ (see Fig. \ref{ai_ai+1} for $n=3$). Thus, substituting $\Delta_i=a_{i+1}$, instead of $\Delta_i=\Delta^*_i$ in (\ref{Pe0_Pe1_dTi}) provides an upper bound on the minimum probability of error.
\begin{figure}
	\centering
	\includegraphics[trim={2cm 3cm 5cm 1.4cm},clip, scale=0.3]{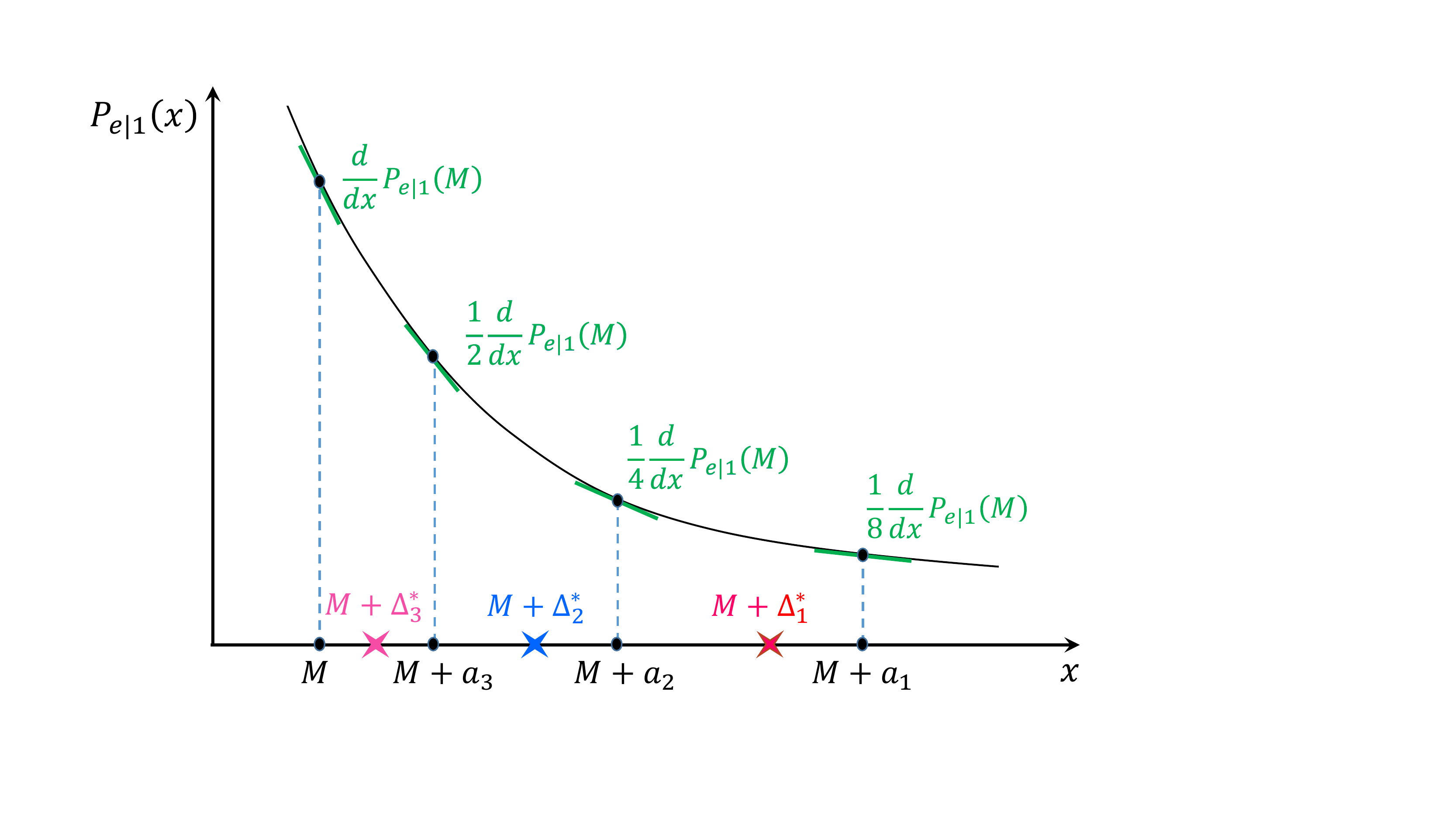}
	\caption{Intervals of optimal number of transmitted molecules in different states}
\label{ai_ai+1}
	\vspace{-0.5em}
\end{figure}
We note that $\Delta^*_i<a_{i}$. Similarly we obtain a lower bound on the minimum probability of error by substituting $\Delta_i=a_{i}$, instead of $\Delta_i=\Delta^*_i$ in (\ref{Pe0_Pe1_dTi}).


\subsubsection{An upper bound on the number of positive release duration increments}
 \label{Bounds for term numbers}
 
As the required memory at the transmitter is proportional to the number of positive release duration increments, we are interested in bounding it. 
\begin{theorem}
	The number of positive release duration increments is upper bounded as 
	\begin{align}
	N<\min_{m \in \{1,\cdots,J\}} \left\{m+\frac{B_M}{a_{J-m+1}} \right\}. \label{Upper_Bound_num_deltas}
	\end{align}
\end{theorem}
\begin{proof}
\emph{Lemma} \ref{lem_1} shows that release duration increments are decreasing. Remind that $C_{1J}$ is an active condition of (\ref{C_{1i}_no isi}) (see Appendix \ref{Appx_Proof_lem_property1}). Thus, from (\ref{C_{1i}_no isi}) and \emph{Lemma} \ref{lem_1}, we have
 \begin{align}
 \left\{
 \begin{array}{rl}
 &\sum_{i=1}^{J} \Delta^*_i=B_M,\\
 &\Delta^*_1>\Delta^*_2>\cdots>\Delta^*_{J-1}>\Delta^*_J>0,
 \end{array} \right. \nonumber
 \end{align}
which results in
$ J<1+\frac{B_M}{\Delta^*_{J-1}}. \nonumber$
 By following a similar approach and considering $\sum_{i=1}^{J-m} \Delta^*_i$ from $C_{1J}$, we obtain
 \begin{align}
N<m+\frac{B_M}{\Delta^*_{J-m}} <m+\frac{B_M}{a_{J-m+1}}. \nonumber
 \end{align}
Note that $\Delta^*_{J-m}>a_{J-m+1}$. Further, from (\ref{min_ai s}), the following equation can be easily obtained.
 \begin{align}
 a_i=-T_\text{th}\text{W}\Bigg(\frac{-\sqrt[T_\text{th}]{\frac{-(T_\text{th}!)}{2^{J-i+1}}\frac{d}{dx}P_{e|1}(x)\Big  |_{M}}}{T_\text{th}}\Bigg)-\lambda-M, \label{ai_closed_formula}
 \end{align}
where $\text{W}(\cdot)$ is the Omega Function defined as $\text{W}(ze^z)=z$.
\end{proof}

\section{Main results in ISI case} \label{sec_main results ISI case}

In this section, we consider the problem in the presence of ISI. First, we assume one-symbol ISI and find the optimal solution for release duration increments to minimize the probability of error. Then, we consider two-symbol ISI.

\subsection{One-symbol ISI} \label{1 mem. - ISI}

If the ISI is one symbol, the received molecules at the receiver have a Poisson distribution whose parameter is a function of the current bit and the previous bit. So if a bit ``0" is transmitted, the ISI is erased and received signal is only a function of the current bit. As a result, if we know the current state, we know the previous bit and we can obtain the probability of error.
From (\ref{poiss_dist}), we have: $
Y_i \sim \text{Poisson}(p_0 X_i+p_1 X_{i-1}+\lambda). \nonumber
$
We are going to find the conditional probability of error.
If we use $P_{e|ji}$	for error probability where the current bit is ``$i$" and the previous bit is ``$j$", we have
\begin{align}
P_{e|B}=\frac{1}{2}(P_{e|0B}+P_{e|1B}), \text{ for } B=0, 1. \nonumber
\end{align}
The previous bit has two cases:\\
I. The previous bit is ``0":
	ISI is zero and we have
	\begin{align}
	P_{e|0B}=P_{e|B}(Bp_0(M+\Delta_1)), \text{ for } B=0, 1. \nonumber 
	\end{align}
II. The last bit is ``1":
	 ISI get different amounts based on the previous state. If the previous state is $s_{j-1}$, the released molecule increment is $\Delta_j$. Since the previous bit is ``1", the current state is $s_{j}$ and the current released molecule increment is $\Delta_{j+1}$. Thus, the average amount of $P_{e|1B}, B\in \{0,1\}$ is
	\begin{align}
	&P_{e|1B}=\sum_{j=1}^{\infty}P(s_{j-1})P_{e|s_{j-1}}\nonumber\\
	&=\sum_{j=1}^{\infty}\frac{1}{2^{j}}P_{e|B}\bigg(Bp_0(M+\Delta_{j+1})+p_1(M+\Delta_j)\bigg).\nonumber
	\end{align}

Our goal is finding $\Delta_i$s to minimize $P_e$.
\begin{align}
\min_{\{\Delta_i\}_{i=1}^\infty}&{P_e}=\frac{1}{2}(P_{e|0}+P_{e|1}), \label{min_pe_isi}\\
\text{s. t. : }
&C_{1i}(\{\Delta_i\}_{i=1}^\infty)= B_M-\sum_{j=1}^{i}\Delta_j\geq 0, \text{ }\forall i,  \label{C1i_isi} \\  
&C_{2i}(\{\Delta_i\}_{i=1}^\infty)= \sum_{j=1}^{i}\Delta_j \geq 0,\text{ }\forall i.    \label{C2i_isi}                   
\end{align}
In Appendix \ref{proof_H_positive semi definite}, for the general case of finite memory channel (\emph{i.e.}, there exists a finite number $K$ where $p_i=0 \text{, for }i>K+1$), we show that the objective function in (\ref{min_pe_isi}) is convex and therefore the optimal solution exists. As a result, this function has a minimum in the compact domain of $\Delta_i$s, which is the optimal solution.

The Lagrangian method results in a problem, which is difficult to solve in a closed form. To simplify this problem, we write the probability of error in (\ref{min_pe_isi}), as a function of ISI value in each state.

In no-ISI case, we defined the states as the number of transmitted bits ``1" after the last bit ``0". 	For the case of one symbol ISI, using the same definition of states, we can determine the value of ISI, $v_{i}$ for $i=1,2,\cdots$, in state $s_{i-1}$, as given in TABLE \ref{table_state_1ISI}. Therefore, the previous definition of states is used in one symbol ISI case, with the probabilities of states given in (\ref{state_probes}), \emph{i.e.}, ${P}(v_i)={P}(s_{i-1})=\frac{1}{2^i}$. As a result, in state $s_{i-1}$, we choose the number of released molecules, $l_i$, as a function of ISI value (\emph{i.e.}, $l_i=M+\Delta_i=f(v_i)$), similar to \cite{movahednasab2016adaptive}.
\begin{table} 
	\centering
	\caption{States and ISI values for one-symbol ISI case}
	\vspace{-0.5em}
	\begin{tabular}{|c|c|}
		\hline State & ISI value \\\hline
		${s}_0=0$ & $v_1=0$\\
		${s}_1=01$ & $v_2=p_1(M+\Delta_1)$\\
		${s}_2=011$ & $v_3=p_1(M+\Delta_2)$ \\
		${s}_3=0111$ & $v_4=p_1(M+\Delta_3)$ \\
		$\vdots$ & $\vdots $ \\
		\hline
	\end{tabular}
	\label{table_state_1ISI}
\vspace{-0.5em}
\end{table}
Thus, the probability of error can be written as
\begin{align}
P_e=\frac{1}{2}\mathbb{E}_I\Big[P_{e|0}(I) +P_{e|1}(p_0f(I)+I)\Big],\label{pe__isi}
\end{align}
where 
$\mathbb{E}_I[X]=\sum_{i=1}^{\infty} {P}(s_{i-1})(X|{I=v_i})$.
Now, we minimize the probability of error, subject to
\begin{align}
&C_{1i}(\{l_i\}_{i=1}^{\infty})=B_M-\sum_{j=1}^{i}(l_j-M)\geq 0,\text{ }\forall i,  \label{C1i_isi_b}  \\  
&C_{2i}(\{l_i\}_{i=1}^{\infty})= \sum_{j=1}^{i}(l_j-M) \geq 0 ,\text{ }\forall i. \label{C2i_isi_b}
\end{align}
To simplify the problem, similar to \cite[p. 247]{movahednasab2016adaptive}, we assume that $I$ is a random variable, which takes values $v_1,v_2,\cdots$ in each state.
Considering this assumption, the first term of (\ref{pe__isi}) is fixed and only the second term depends on $l_i$s. Therefore, we can simplify the optimization problem and obtain a sub-optimal solution by minimizing only the second term of (\ref{pe__isi}) as follows.
\begin{align}
\min_{\{l_i\}_{i=0}^{\infty}} &\sum_{j=1}^{\infty} \frac{1}{2^{j}} \sum_{y< T_\text{th} }^{}e^{-(p_0l_i+v_i+\lambda)}\frac{(p_0l_i+v_i+\lambda)^y}{y!},\label{min_pe_isi_b}
\end{align}
subject to (\ref{C1i_isi_b}) and (\ref{C2i_isi_b}).
Using Lagrangian method, at the optimal solution, we have
\begin{align}
\nabla\mathbb{E}_I[P_{e|1}(p_0l_i^*+v_i^*)]  =\sum_{k=1}^{2}\sum_{i=1}^{\infty}\mu^*_{ki} \nabla C_{ki}(\{l_i^*\}_{i=1}^{\infty}),
\label{kkt_isi_b}
\end{align}
where $\mu^*_{ki}$s are Lagrangian multipliers. Due to KKT conditions, we have $\mu^*_{ki} C_{ki}=0, \forall i=1,\cdots$.
In Appendix \ref{appx_sol_1bit_ISI}, we obtain a local minimum for (\ref{min_pe_isi_b}), which is a global minimum (because $P_{e|1}$ is a convex function and the domain of $l_i$s in conditions (\ref{C1i_isi_b}) and (\ref{C2i_isi_b}) is compact). Since $l^*_i=M+\Delta^*_i$, the $\Delta^*_i$s are obtained from (\ref{li*s_appx}) as
\begin{equation}
\Delta_i^*=\frac{1}{p_0}
\begin{cases}
p_1M+\Delta_{1, \text{no-ISI}}^*,\text{ }i=1 \\
\Delta_{i, \text{no-ISI}}^*-p_1\Delta_{i-1}^*,\text{ }2\leq \text{ }i\leq J \\
0, \text{ }i> J.
\end{cases} \label{subopt dT_i_mem1}
\end{equation}


\subsection{Two-symbol ISI} \label{2 mem. - ISI}

We know the storage is filled when a bit ``0" comes. 
We use different release duration for each bit ``1" after bit ``0". For example, if the number of transmitted ``1"s after the last ``0" is $i$ bits, the released molecule increment for the next bit ``1" is $\Delta_{i+1}$.
We update the definition of state diagram to match two-symbol ISI. We define the states as $0\underbrace{1\cdots1}_{i+1}$ or $0\underbrace{1\cdots1}_{i}0$, where $i=0,1,\cdots$, as shown in Fig. \ref{f12_markovchain_2mem}. These states show the sequence of previous transmitted bits (at least two bits) and include sufficient information about the ISI values. TABLE \ref{table1_state_2ISI} shows ISI values for these states. For example, if the system is in state $\tilde{s}_3=010$ and bit ``1" comes, ISI takes value of $p_1(M+\Delta_1)$ and the system goes to state $\tilde{s}_2 =01$.
 \begin{table} 
 	\centering
 	\caption{States and ISI values for two-symbol ISI case}
 	\vspace{-0.5em}
 	\begin{tabular}{|c|c|}
 		\hline State & ISI value \\\hline
 		$\tilde{s}_1=00$ & $v_1=0$\\
 		$\tilde{s}_2=01$ & $v_2=p_1(M+\Delta_1)$\\
 		$\tilde{s}_3=010$ & $v_3=p_2(M+\Delta_1)$ \\
		$\tilde{s}_4=011$ & $v_2=p_1(M+\Delta_2)+p_2(M+\Delta_1)$ \\
 		$\tilde{s}_5=0110$ & $v_5=p_2(M+\Delta_2)$ \\
 		$\tilde{s}_6=0111$ & $v_6=p_1(M+\Delta_3)+p_2(M+\Delta_2)$ \\
  		$\vdots$ & $\vdots $ \\
  		\hline
 	\end{tabular}
 	\label{table1_state_2ISI}
 \end{table}

The transition matrix of Fig. \ref{f12_markovchain_2mem} is $\tilde{\Pi} =[\pi_{ij}]$, where
\begin{equation}
\pi_{ij}=
\begin{cases}
1/2, \text{ }(i,j) \in C_\pi \nonumber\\  
	0, \text{ otherwise},        \nonumber
\end{cases} \nonumber
\end{equation}
and 
$C_\pi=\{(2k-1,1),(2k-1,2),(2k,2k+1),(2k,2k+2)|\forall k\in \mathbb{N}\}. \nonumber$
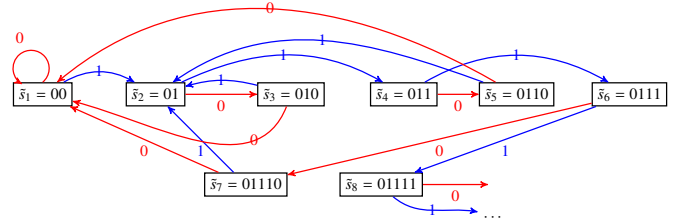
\begin{figure}
	\begin{center}
		\vspace{-1.8em}
		\scalebox{0.6}{%
			\begin{tikzpicture} [
			point/.style={coordinate},>=stealth',thick,draw=black,
			tip/.style={->,shorten >=0.007pt},
			text height=1.5ex,text depth=.25ex 
			]
			\path (2.5,6.5) node[draw,shape=rectangle] (A) {$\tilde{s}_1=00$};  
			\path (5,6.5) node[draw,shape=rectangle] (B) {$\tilde{s}_2=01$};
			\path (8,6.5) node[draw,shape=rectangle] (C) {$\tilde{s}_3=010$};
			\path (10.5,6.5) node[draw,shape=rectangle] (D) {$\tilde{s}_4=011$};  
			\path (13,6.5) node[draw,shape=rectangle] (E) {$\tilde{s}_5=0110$};
			\path (15.5,6.5) node[draw,shape=rectangle] (F) {$\tilde{s}_6=0111$};
			\path (7,4.5) node[draw,shape=rectangle] (G) {$\tilde{s}_7=01110$};
			\path (10,4.5) node[draw,shape=rectangle] (H) {$\tilde{s}_8=01111$};  
			\path (12.5,3.8)	node (V) {$\cdots$}; 
			\path (12.5,4.5)	node (W) {}; 
			\draw [red,->] (B) -- node [below, midway] {$0$} (C);
			\draw [red,->] (D) -- node [below, midway] {$0$} (E);
			\draw [red,->] (H) -- node [below, midway] {$0$} (W);
			\draw [red,->] (F) -- node [below, midway] {$0$} (G);
			\draw [blue,->] (F) -- node [below, midway] {$1$} (H);
			\draw [blue,->] (G) -- node [below, midway] {$1$} (B);
			\draw [red,->] (G) -- node [below, midway] {$0$} (A);
			\draw [blue,->, bend left] (D) [out=30 , in=150] to node {$1$} (F);
			\draw [blue,->, bend left] (C) [out=345 , in=195] to node {$1$} (B);
			\draw [blue,->, bend left] (E) [out=340 , in=215] to node {$1$} (B);
			\draw [red,->, bend left] (E) [out=330 , in=220] to node {$0$} (A);
			\draw [blue,->, bend left] (A) [out=30 , in=150] to node {$1$} (B);
			\draw [blue,->, bend left] (H) [out=330 , in=181] to node {$1$} (V);
			\draw [blue,->, bend left] (B) [out=25 , in=150] to node {$1$} (D);
			\draw [red,->, bend left] (C) [out=70 , in=167] to node {$\text{ }\text{ }\text{ }\text{ }\text{ }\text{ }\text{ }\text{ }\text{ }\text{ }\text{ }\text{ }\text{ }\text{ }\text{ }\text{ }\text{ }\text{ }\text{ }\text{ }0$} (A);
			\draw [red,->] (A.90) arc (-50:240:4mm) node[pos=0.5,above left] {$0$} (A);
			\end{tikzpicture}
		}
	\vspace{-2em}
	\end{center}
	\caption{{Transition diagram of two-symbol ISI}} 
	\label{f12_markovchain_2mem}
\end{figure}
The steady-state distribution $P_{\tilde{s}}=[P_{\tilde{s}_1},P_{\tilde{s}_2},\cdots]$ is derived as $P_{\tilde{s}}=P_{\tilde{s}}\tilde{\Pi}$, which results in
\begin{align}
P_{\tilde{s}}=[{1}/{4},{1}/{4},{1}/{8},{1}/{8},{1}/{16},{1}/{16},\cdots].  
\label{state_dist_2ISI}
\end{align}

As shown in Fig. \ref{f12_markovchain_2mem}, in states $\tilde{s}_1,\tilde{s}_3,\tilde{s}_5,\cdots$, the previous transmitted bit is ``0"; Thus, the storage is full and we choose  released molecule increment $\Delta_1$ in these states (if a bit ``1" comes in these states, the released number of molecules are $l_1=l_3=l_5=\cdots=M+\Delta_1$). In states $\tilde{s}_2,\tilde{s}_4,\tilde{s}_6,\cdots$, the number of transmitted bits ``1" after the last bit ``0" are respectively $1,2,3,\cdots$; Thus, we choose released molecule increments $\Delta_2,\Delta_3,\Delta_4,\cdots$ in these states, respectively (\emph{i.e.}, $l_2=M+\Delta_2, l_4=M+\Delta_3,,l_6=M+\Delta_4,\cdots$).
Similar to the case of one-symbol ISI, again we assume that $I$
is a random variable, which takes values $v_1,v_2,\cdots$ in states $\tilde{s}_1,\tilde{s}_2,\cdots$ and released number of molecules in these states are a function of ISI (\emph{i.e.}, $l_i=f(v_i)$).
Using these assumptions, we can write the probability of error as (\ref{pe__isi}) and only minimize the second term (the first term is fixed) as:
\begin{align}
&\mathbb{E}_I[P_{e|1,I}]=\sum_{i=1}^{\infty} P_{\tilde{s}_i} P_{e|1}(p_0l_i+v_i),\nonumber
\\
&=\sum_{i=1}^{\infty} \left(P_{\tilde{s}_{2i-1}} P_{e|1}(p_0l_1+v_{2i-1})+P_{\tilde{s}_{2i}}  P_{e|1}(p_0l_{2i}+v_{2i})\right),
\label{pe_isi_2b}
\end{align}
to obtain a sub-optimal solution.
In Appendix \ref{appx_sol_2bit_ISI}, we obtain a local minimum (with optimal released numbers $\{l^*_i\}_{i=0}^{\infty}$), for (\ref{pe_isi_2b}), subject to the same conditions as in (\ref{C1i_isi_b}) and (\ref{C2i_isi_b}).
It is also a global minimum (because $P_{e|1}$ is a convex function and the domain of $l_i$s  is compact). Since $l^*_i=M+\Delta^*_i$, the $\Delta^*_i$s are obtained from (\ref{li*s_appx_2bit_isi}) as
\begin{equation}
\Delta_i^*=\frac{1}{p_0}
\begin{cases}
(p_1+\frac{p_2}{2})M+\Delta_{1, \text{no-ISI}}^*,\text{ } i=1, \\
p_2M+\Delta_{2, \text{no-ISI}}^*-p_1\Delta_1^*, \text{ }i=2,\\
\Delta_{i, \text{no-ISI}}^*-p_1\Delta_{i-1}^*-p_2\Delta_{i-2}^*, \text{ }2<i \leq J, \\
0,\text{ }i> J.
\end{cases} \label{subopt_dTi_mem2}
\end{equation}

\subsection{Adaptive threshold receiver in the presence of ISI}\label{Adaptive_th_isi}

So far, we consider a fixed threshold receiver and obtain the optimal release durations in each state. These optimal release durations can be used to update the threshold in each state. Similar to the channel without ISI, we use ML receiver and choose the best threshold in state $s_{j-1}$, as:
\begin{align}
\frac{P(y|s_{j-1})}{P(y|s_0)}  \underset{0}{\overset{1}\gtrless }1,   \nonumber
\end{align}
and thus, 
$
\Big(\frac{M+p_0\Delta_j+p_1\Delta_{j-1}+p_2\Delta_{j-2}+\lambda}{p_1(M+\Delta_{j-1})+p_2(M+\Delta_{j-2})+\lambda}\Big)^{T^j_\text{th}} \underset{0}{\overset{1}\gtrless } e^{p_0(M+\Delta_j)}. \nonumber
$
Therefore, threshold is
\begin{align}
T^j_\text{th} =\frac{p_0(M+\Delta_j)}{\text{ln}\Big(1+\frac{p_0(M+\Delta_j)}{p_1(M+\Delta_{j-1})+p_2(M+\Delta_{j-2})+\lambda}\Big)}. \label{Adap_Th_isi}
\end{align}
The receiver which save these adaptive thresholds in its memory and uses $T^j_\text{th}$ in state $s_{j-1}$, is called adaptive threshold receiver.


\section{Numerical Results} \label{sec_numerical results}

In this section, we provide simulations to show the performance improvement achieved by the proposed adaptive release duration scheme. 
We choose a transmitter with $\beta=2$molecule/sec, $T=25$sec and $T_M=4$sec, which are the molecule production rate, slot duration and non-adaptive release duration, respectively. We assume hitting probabilities for one-symbol ISI case as $p_0=0.9$ and $p_1=0.1$ and for two-symbol ISI case as $p_0=0.85$, $p_1=0.1$ and $p_2=0.05$.
First, we compare the performance of four different strategies in terms of probability of error for no-ISI case.
\\
$\bullet$ {Strategy 1 (non-adaptive strategy):} release duration is fixed and the receiver uses a fixed threshold (probability of error is derived in (\ref{Pe_fixed})).\\
$\bullet$ {Strategy 2:} 
release durations are optimal and the receiver uses a fixed threshold (probability of error is discussed in \emph{Remark}~\ref{remark:min Pe_ opt delta & fixed th}).\\
$\bullet$ {Strategy 3:}
 release durations are optimal and the receiver uses these optimal release durations to obtain adaptive thresholds numerically in one iteration such that $P_e$ decreases (probability of error is derived in Appendix \ref{Adaptive_th_no_isi},  (\ref{pe_adaptive_adaptive Th_1iter})).\\
$\bullet$ {Strategy 4:}
 release durations and thresholds are optimized simultaneously, \emph{i.e.}, in several iterations such that the probability of error converges (probability of error is discussed in Appendix \ref{Adaptive_th_no_isi}, \emph{Remark} \ref{remark:min Pe_opt delta & opt th _ simmultaneously}).
\begin{figure}
	\centering
	\includegraphics[trim={0cm 0.2cm 0cm 1cm},clip,scale=0.33]{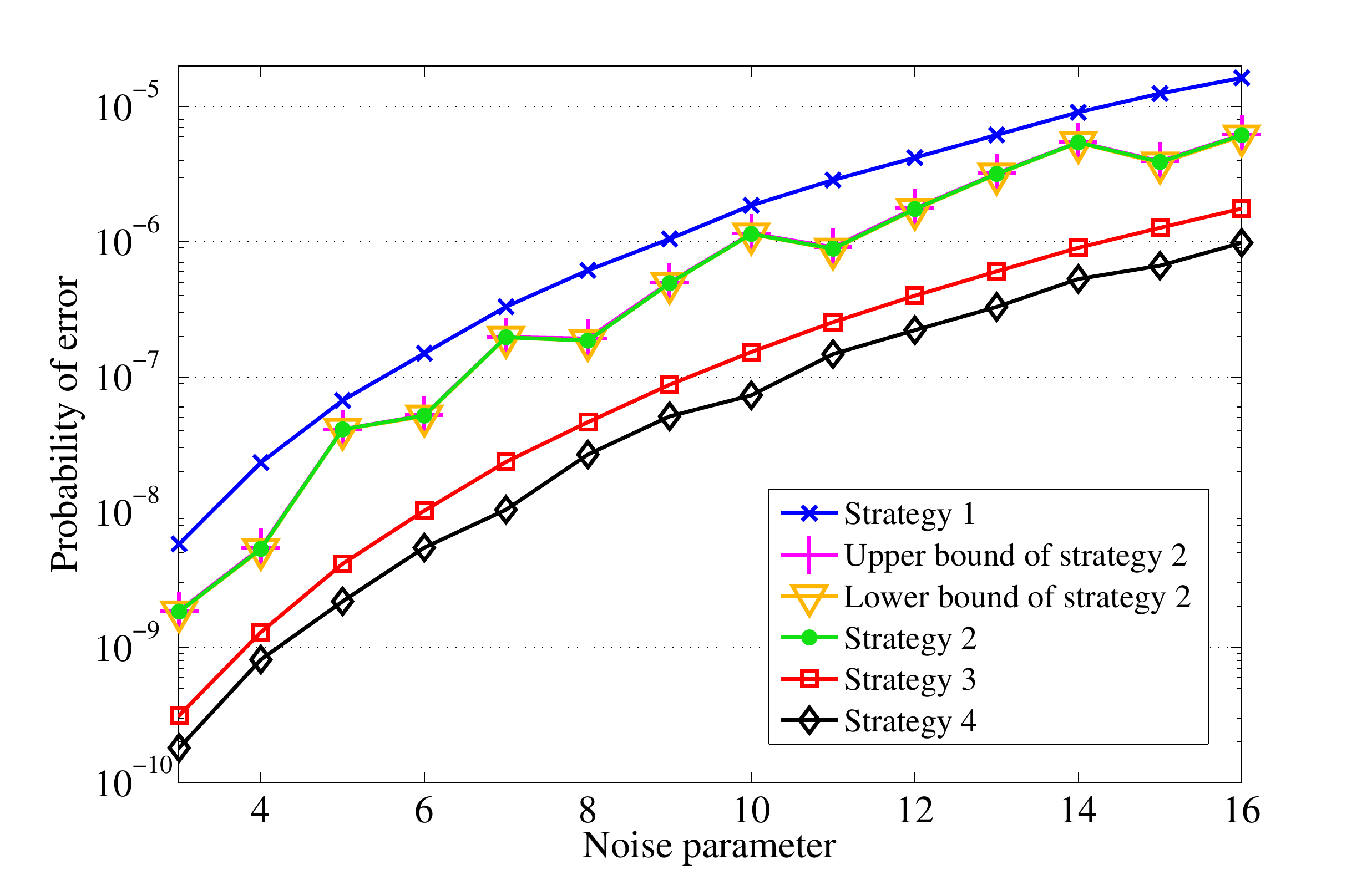}
	\caption{Error probability versus noise parameter for no-ISI case}
	\label{figs_no_isi_method3_Ln}
	\vspace{0em}
\end{figure}
\begin{figure}
	\centering
	\includegraphics[trim={0cm 0.2cm 0cm 1cm},clip,scale=0.33]{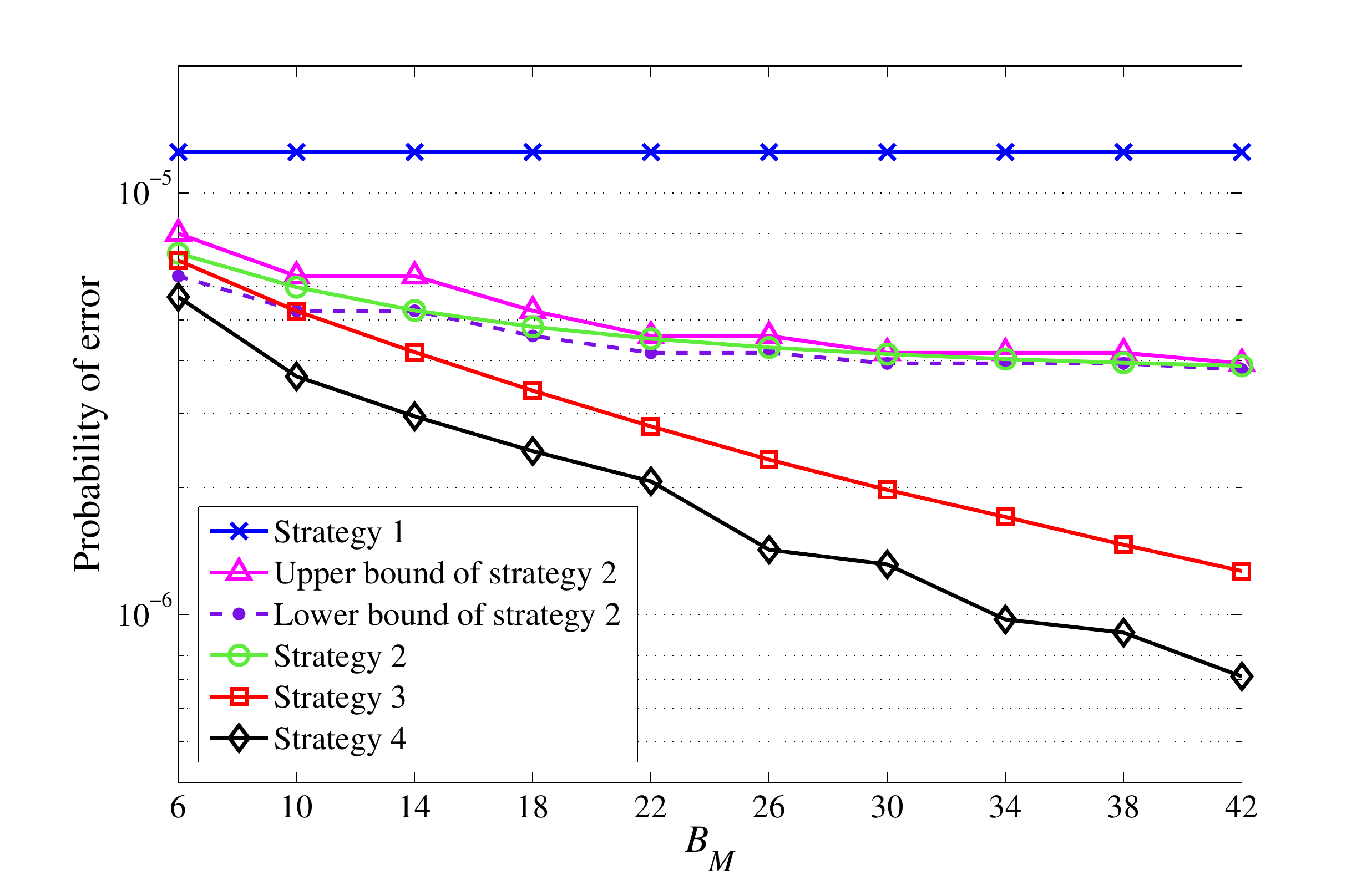}
	\caption{Error probability versus storage capacity for no-ISI case with $\lambda=15\text{ molecules}$ }
	\label{figs_no_isi_method3_BM}
	\vspace{-1em}
\end{figure}

Fig. \ref{figs_no_isi_method3_Ln} shows the probability of error versus the noise parameter for the mentioned strategies. Since finding optimal release duration increments needs to solve the optimization problem in (\ref{minF}), upper and lower bounds in Section~\ref{Upper_Lower bounds for Pe} respectively are plotted without solving this optimization problem and only by using the interval of optimal solution. It can be seen that these two bounds are very close to each other.
As one can see, strategy 1, which is the simplest technique for transmission and reception, can be improved by changing the release durations and/or thresholds. Though the complexity of the system increases compared to non-adaptive strategy, we have shown that the transmitter and receiver need only a limited memory to save the amounts of release duration increments and adaptive thresholds.
For example, there is nine positive optimal release duration increments that the transmitter save these numbers in its memory. Also the number of adaptive threshold that the receiver save in its memory is nine numbers. 
  
Fig. \ref{figs_no_isi_method3_BM} shows the probability of error versus the storage capacity for the above  four strategies. As one can see, in strategy 1, the storage capacity has no effect on the system performance while for three other methods, increasing storage capacity improves the system performance. As shown in Fig.~\ref{figs_no_isi_method3_BM}, for strategies 3 and 4 (which we change both release durations and thresholds), with increasing storage capacity, $P_e$ decreases exponentially. This figure shows that changing the release durations and thresholds (strategies 3 and 4) significantly improves the probability of error.

In Fig. \ref{figs_no_isi_Upper_Bound_number_of_deltas_BM}, we show the upper bound on the number of positive release duration increments derived in (\ref{Upper_Bound_num_deltas}) versus the storage capacity for noise parameters $\lambda=3,7,11,15$. 
For a fixed $B_M$, we observe that by increasing $\lambda$, the upper bound decreases. The reason is due to the fact that in higher noise regimes ($\lambda=15$), the number of required molecules is high, which translates to larger release duration; Therefore, there is less time remained to increase the release durations for adaptive strategy.
Similarly for a larger storage capacity, the stored molecules are more and we need smaller release duration ($T_M$) in non-adaptive strategy. Thus, we have more freedom to allocate the remaining time (\emph{i.e.}, $T-T_M$) for adaptive strategy.
we can use this time in more states and thus the number of positive optimal release duration increments increases.
\begin{figure}
	\centering
	\includegraphics[trim={0cm 0.1cm 0cm 1cm},clip,scale=0.33]{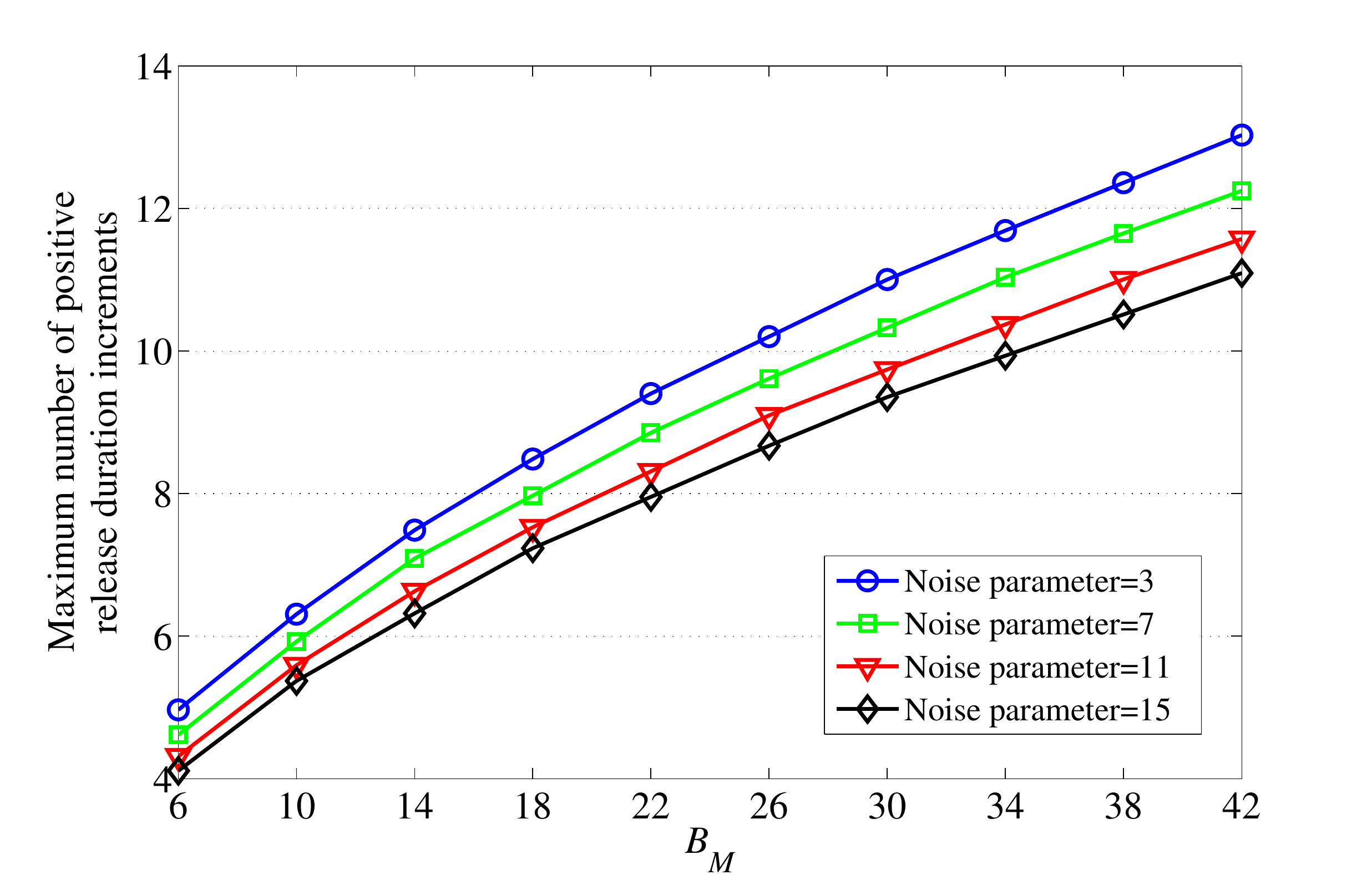}
	\caption{Upper bounds on the number of positive release duration increments versus storage capacity}
	\label{figs_no_isi_Upper_Bound_number_of_deltas_BM}
	\vspace{-1em}
\end{figure}

For the channel with ISI, a method is proposed in \cite{movahednasab2016adaptive} for a receiver with a fixed threshold while the transmitter has no limitation. It shows that the optimal transmission rates are chosen such that the received signal for bit ``1" has a fixed rate for almost all values of ISI. For a transmitter with limited molecule production rate and storage capacity, when ISI is zero ($I=0$), the received rate cannot be more than $p_0M$. Thus, the maximum fixed rate of molecules that can be received is
\begin{align}
p_0X_i+I=p_0M. \label{fixed_rate_movahed}
\end{align}
 We simplify this method for our transmitter with limitations in one-symbol ISI case and find a closed form for the transmission rates (see Appendix \ref{movahed_1mem}).
Now, we compare strategies 1, 4 and the following strategies:\\
$\bullet$ {Strategy 5:} The method proposed in \cite{movahednasab2016adaptive}; release durations are such that the received signal has a fixed rate and the receiver uses a fixed threshold.\\
$\bullet$ {Strategy 6 (sub-optimum strategy):} release durations are sub-optimum as in (\ref{subopt dT_i_mem1}) and (\ref{subopt_dTi_mem2}) for one and two-bit ISI, respectively. The receiver uses adaptive thresholds in (\ref{Adap_Th_isi}).\\

Fig. \ref{f_shabih_isi_1mem} shows the probability of error versus the noise parameter for these four strategies in a channel with one-symbol ISI. For strategy 5, we plot the upper and lower bounds on the probability of error (derived in (\ref{movahed_upper B_1mem}) and (\ref{movahed_lower B_1mem}), respectively). It can be seen that strategy 1 is better than strategy 5. Because, strategy 5 forces the transmission scheme to have a fixed rate in (\ref{fixed_rate_movahed}) at the receiver. Thus, it forces the transmitter to decrease release rate (\emph{i.e.}, $X_i<M$) when ISI is not zero. It means that the transmitter uses less number of molecules to send bit ``1" and thus the probability of error increases. Our proposed strategy in both optimum (strategy 4) and sub-optimum cases (strategy 6) significantly outperform strategies 1 and 5, because strategy 5 is not proposed for a system with transmitter's limitations (limited molecule production rate and storage capacity). Therefore, considering these limitations in designing the modulation scheme, can significantly improve the performance. It can be observed that the sub-optimum strategy 6 converges to the optimum strategy 4 as noise increases.
\begin{figure}
	\centering
	\includegraphics[trim={0cm 0.2cm 0cm 0.5cm},clip,scale=0.33]{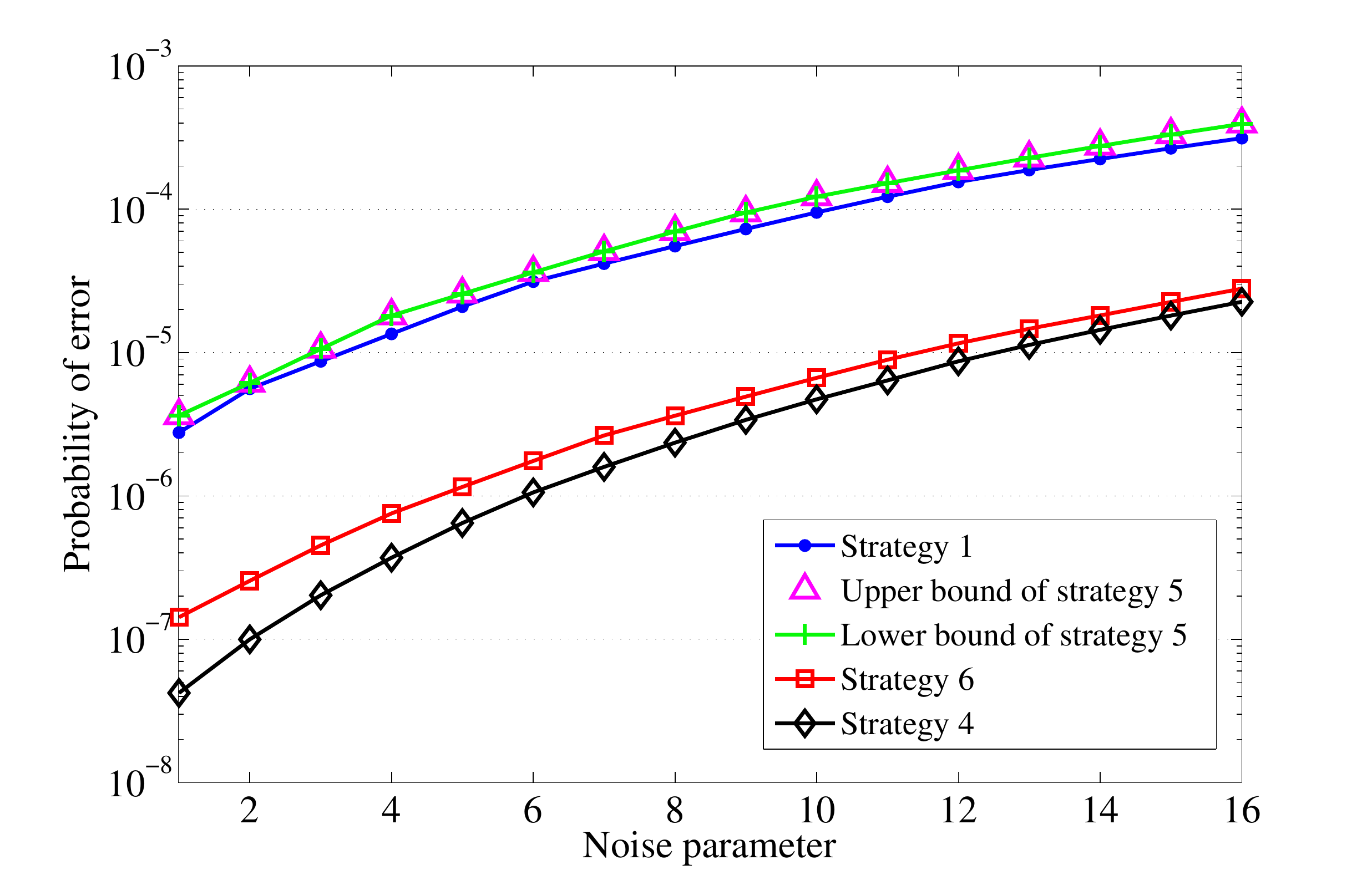}
	\caption{Error probability versus noise parameter for one-symbol ISI case}
	\label{f_shabih_isi_1mem}
	\vspace{-1em}
\end{figure}

In case of two-symbol ISI, similar results can be seen in  Fig. \ref{f_shabih_isi_2mem}, where the error probability of strategy 5 is plotted via simulations because there is not a closed form for the number of released molecules. We compare strategies 1 and 5 with our optimum (strategy 4) and sub-optimum (strategy 6) strategies. We see that our strategies improve the system performance in terms of probability of error, compared to strategies 1 and 5.
\begin{figure}
	\centering
	\includegraphics[trim={0cm 0.2cm 0cm 0.5cm},clip,scale=0.33]{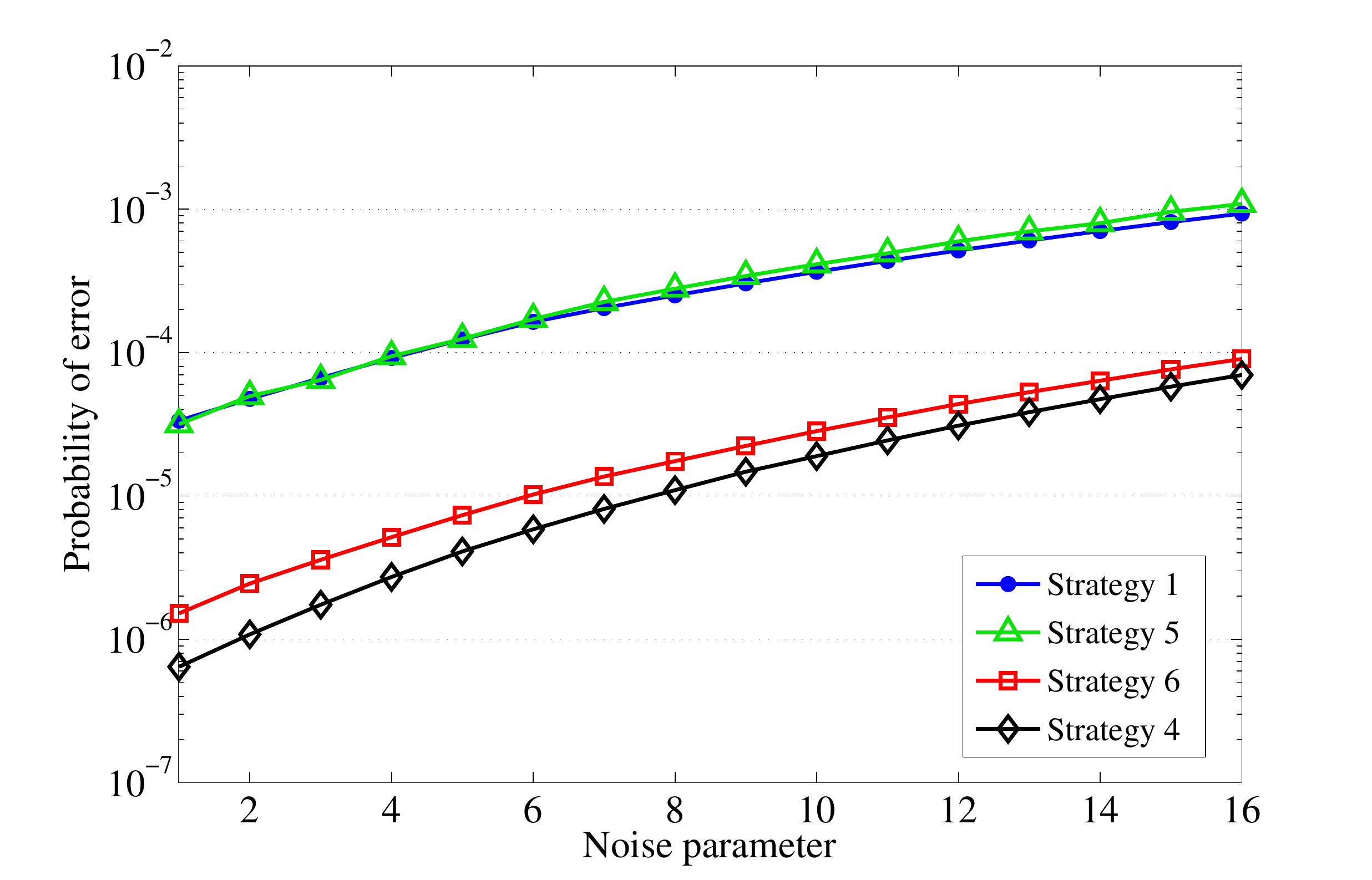}
	\caption{Error probability versus noise parameter for two-symbol ISI case}
	\label{f_shabih_isi_2mem}
	\vspace{-1em}
\end{figure}


\section{Conclusion} \label{conclusion}

In this paper, the transmitter's limitations are studied in MC, in particular the limited molecule production rate and finite storage capacity. By changing the duration in which molecules are released, an adaptive release duration modulation is proposed. The objective is finding the optimal transmission release durations to minimize the probability of error.
Important properties of the optimal solution are proved and it is shown that optimal release durations are decreasing. Using these properties, upper and lower bounds are derived on the minimum probability of error. It is shown that increasing the storage capacity significantly improves the system performance in terms of probability of error. It is also shown that an adaptive threshold receiver performs better than a fixed threshold receiver. 
The proposed modulation scheme also performs better than transmission method proposed in  \cite{movahednasab2016adaptive}. Because this method is not  proposed for a system with transmitter's limitations.
While in our proposed strategy, the complexity of the system is increased compared to non-adaptive strategy, we show that the required memories for the proposed transmitter and receiver are limited.


\begin{appendices}

\section{Convexity of probability of error in (\ref{def_F(delta_i)})} 
\label{Convex objective function}
	
To prove the convexity of $P_{e|1}$, consider $F(\cdot)$ in (\ref{def_F(delta_i)}). We have
$\frac{d^2}{d\Delta_i d\Delta_j}F(\Delta_1, \Delta_2, \cdots)=0$ for $i\neq j$ and
$
\frac{d^2}{d{\Delta_i}^2}F(\Delta_i)=\sum_{i=1}\frac{1}{2^i}g(\Delta_i),
$
where $g(.)$ is defined as
\begin{align}
g(x) =e^{-y}\frac{y^{T_\text{th}-1}}{T_\text{th}!}(y-T_\text{th})\Big|_{y=M+x+\lambda}. \nonumber
\end{align}
Since $M+x+\lambda>T_\text{th}$, $g(\Delta_i)>0, \forall i=1,\cdots,J$. Thus, $F(\cdot)$ is convex. Therefore, $P_{e|1}$ is also a convex function.

\section{Proof of Lemma \ref{lem_1}} \label{Appx_Proof_lem_property1}

We propose a solution for (\ref{minF}) that satisfies  all conditions in (\ref{C_{1i}_no isi}) and (\ref{C_{2i}_no isi}). Assume only $C_{1i}=0$, for $i \geq J$ ($C_{2i} \neq 0$) and $\mu^*_{1i}=0$ for $i >J$. Thus, the conditions result in
\begin{equation}
\begin{cases}
     \beta(T-T_M)-\sum_{j=1}^{i}\Delta^*_j > 0, \text{ for }i<J, \\
     \beta(T-T_M)-\sum_{j=1}^{J}\Delta^*_j= 0,\\
     \Delta^*_j= 0,\text{ } \mu^*_{1j}=0, \text{ for } j>J.
\end{cases}
\label{positive_sum}
\end{equation}
First, note that a solution subject to the conditions in (\ref{positive_sum}) satisfies the conditions in (\ref{C_{1i}_no isi}) and (\ref{C_{2i}_no isi} ). 
The  above assumptions and (\ref{g1}) result in
\begin{align}
\left(
{\frac{\partial F}{\partial \Delta_1}}, { \frac{\partial F}{\partial \Delta_2}}, 
\ldots \right)^\text{T} &=\mu^*_{1J}\left( -1, \ldots, -1, 0, \ldots \right)^\text{T}, \label{g2}
\end{align}
where $\mu^*_{1J}>0$ due to KKT conditions. Since the objective function in (\ref{minF}) is convex and solution of (\ref{positive_sum}) satisfies the conditions in (\ref{C_{1i}_no isi}) and (\ref{C_{2i}_no isi}), this solution is the global minimum.
From (\ref{g2}), we obtain
\begin{align}
\frac{d}{{d\Delta_k}}P_{e|1}(x)\Big|_{M+\Delta^*_k}=2^{k-h}\frac{d}{{d\Delta_h}}P_{e|1}(x)\Big|_{M+\Delta^*_h}. \label{mu_F1_F2}
\end{align}
Since $P_{e|1}$  is a convex function, for $k>h$, (\ref{mu_F1_F2}) results in $\Delta^*_k < \Delta^*_h$,
which means that the optimal release duration increments are decreasing.
As a result if one of them ($\Delta^*_m$) is negative, then the $\{\Delta^*_i\}_{i=m,\cdots,J}$ are negative too. (\ref{positive_sum}) results in
\begin{align}
\beta(T-T_M)-\sum_{j=1}^{J-1}\Delta^*_j= \Delta^*_J<0, \nonumber
\end{align}
which is in contradiction with (\ref{positive_sum}). Thus, all $\Delta^*_j$s are nonnegative.


\section{Proof of Lemma \ref{lem_2}} \label{Appx_Proof_lem_property2}

According to \emph{Lemma} \ref{lem_1}, the optimal release duration increments are non-negative and decreasing and the optimal solution satisfies (\ref{positive_sum}). Assume that $C_{1J}$ is active. Thus
$\sum_{j=1}^{J}\Delta^*_j=\beta(T-T_M)$, 
and $\Delta^*_j=0 \text{, for }j>J$.
If the number of positive $\Delta^*_i$s ($J$) is unlimited, (\ref{positive_sum}) forces
\begin{align}
\lim_{j\rightarrow \infty}\Delta^*_j=0, \label{lim_delta_i}
\end{align} 
and $\Delta^*_{j-1}>0$.
From (\ref{def_F(delta_i)}) and (\ref{g2}), we obtain
\begin{align}
e^{-\Delta^*_{j-1}}\frac{(M+\Delta^*_{j-1}+\lambda)^{T_\text{th}}}{T_\text{th}!}=\frac{e^{-\Delta^*_{j}}}{2}\frac{(M+\Delta^*_{j}+\lambda)^{T_\text{th}}}{T_\text{th}!}. \nonumber
\end{align}
From (\ref{lim_delta_i}) and the above equation, we have
$e^{-\Delta^*_{j-1}}=\frac{1}{2}(\frac{M+\lambda}{M+\Delta^*_{j-1}+\lambda})^{T_\text{th}} < \frac{1}{2},$
which means that $\Delta^*_{j-1}>\ln{2}$.
Thus, $\lim_{i\rightarrow \infty}\Delta^*_{j-1}\neq 0$. This contradicts (\ref{lim_delta_i}) and completes the proof.


\section{Proof of Lemma \ref{lem_3}} \label{Appx_Proof_lem_property3}

Consider the smallest optimal positive released molecule increment as $\Delta^*_J$. Remind that $\Delta^*_j=0$ for $j>J$. We prove that $\Delta^*_J \in [0,a_J]$.
We use contradiction. Suppose  $\Delta^*_J \notin [0,a_J]$ \emph{i.e.}, $\Delta^*_J>a_J$. This scheme is referred as scheme A.
\\
$\bullet$ Scheme A: Choose the smallest optimal positive $\Delta_J=\Delta^*_J$ in state $s_J$ which $\Delta^*_J>a_J$ and $\Delta_j=0$ in state $s_j$ for $j>J$. So condition $C_{1J}$ is active:
$\sum_{j=1}^{J-1}\Delta^*_j +\Delta^*_{J} =\beta(T-T_M)$.
The probability of error for scheme A is
\begin{align}
P_{e}^{A}=\frac{1}{2}\Big( P_{e|0}+
\sum_{j=1}^{J}{\frac{1}{2^{j}}P_{e|1}(M+\Delta^*_j)}+\sum_{j\geq J+1}^{}{\frac{1}{2^{j}}P_{e|1}(M)}\Big).
\nonumber
\end{align}
Now, we show that there is a better scheme of transmission which has a less probability of error.
\\
$\bullet$ Scheme B: Consider real numbers $b_1,b_2 \in \mathbb{R}$ (will be derived later in text). Choose $\Delta_J=b_2$ in state $s_J$, $\Delta_{J+1}=b_1$ in state $s_{J+1}$ and $\Delta_j=0$ in state $s_j$ for $j>J+1$. So condition $C_{1{J+1}}$ is active
$\sum_{j=1}^{J-1}\Delta^*_j +b_1+b_2 =\beta(T-T_M)$.
The probability of error for scheme B is
	\begin{align}
	P_{e}^{B}=&\frac{1}{2} P_{e|0}+\sum_{j=1}^{J-1}{\frac{1}{2^{j}}P_{e|1}(M+\Delta^*_j)}+\frac{1}{2^{J+1}}P_{e|1}(M+b_2)\nonumber\\
	&+\frac{1}{2^{J+2}}P_{e|1}(M+b_1)+\sum_{j\geq J+2}^{}{\frac{1}{2^{j+1}}P_{e|1}(M)}.
	\nonumber
	\end{align}
It is shown in Appendix  \ref{Appendix_Proof_b1_b2}
that positive numbers $b_1$ and $b_2$ exist, which satisfy 
$\Delta^*_{J}=b_1+b_2$ and the following property (see Fig. \ref{f8_b1b2}).
\begin{align}
\frac{d}{{dx}}P_{e|1}(x)\Big|_{M+b_1}=2\frac{d}{{dx}}P_{e|1}(x)\Big|_{M+b_2}
\label{b1,b2 dPe/d eq.}.
\end{align}
The result of subtraction is
\begin{align}
P_{e}^{A}-P_{e}^{B}=&\frac{1}{2^{J+1}}\Big(P_{e|1}(M+\Delta^*_{J})-P_{e|1}(M+b_2)\nonumber\\
&+\frac{1}{2}(P_{e|1}(M)-P_{e|1}(M+b_1)\Big)\nonumber\\
\overset{(a)}{\geq} &\frac{1}{2^{J+1}}\Big(-(\Delta^*_{J}-b_2)\frac{d}{dx}P_{e|1}(M+b_2)\nonumber\\
&+\frac{1}{2}b_1\frac{d}{dx}P_{e|1}(M+b_1)\Big), \nonumber
\end{align}
where (a) follows from the convexity of the probability of error (see Appendix \ref{Convex objective function}) and its decreasing property (see (\ref{d/dx P_e1})). Combining (\ref{b1,b2 dPe/d eq.}) and the above equation, results in
$P_{e}^{A}-P_{e}^{B}>~0$,
which means the error probability of scheme B is less than scheme A. It means scheme A cannot be optimal. So, $\Delta^*_J \in [0,a_J]$ and from (\ref{g2}), $\Delta^*_j\in [a_{j+1},a_j]$ for $j=1,\cdots,J-1$.
\begin{figure}
\vspace{-1em}
	\centering
	\includegraphics[trim={1cm 1.8cm 5cm 1.4cm},clip, scale=0.3]{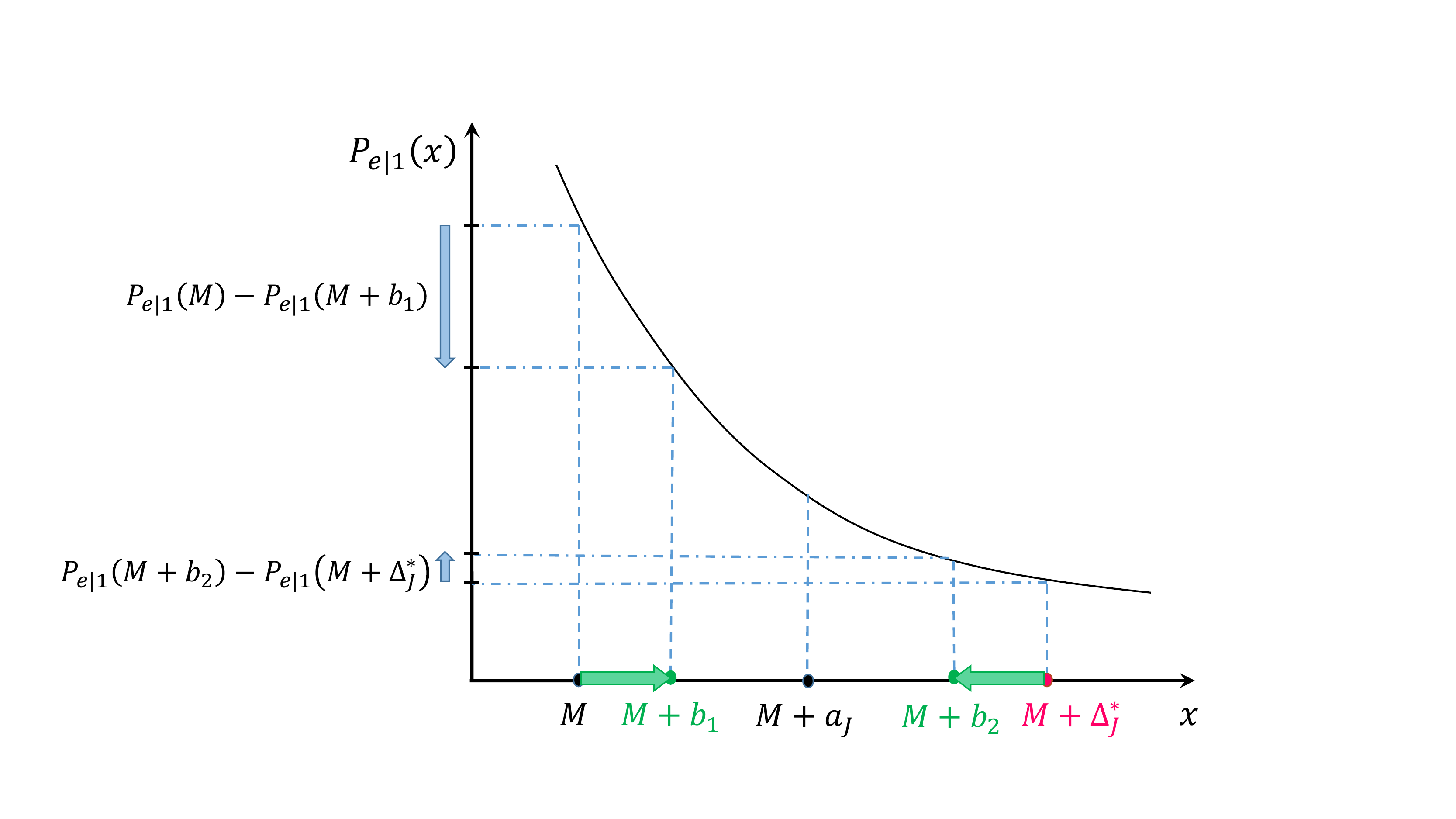}
	\caption{the intervals of optimal released molecule increments}
	\label{f8_b1b2}
	\vspace{-0.5em}
\end{figure}


\subsection{Proof of Existence of Points $b_1$  and  $b_2$} \label{Appendix_Proof_b1_b2}

Suppose that $\Delta^*_J>a_J$. Let us define
	\begin{align}
	h(x)=\frac{d}{dx}P_{e|1}(x)\Big|_{M+\Delta^*_J-x} - \frac{1}{2}\frac{d}{dx}P_{e|1}(x)\Big|_{M+x}.\nonumber
	\end{align}
	We want to show that $h(x)$ has a root in the interval  $[0,\Delta^*_J-~a_J]$. Consider two points $x_1=0$ and $x_2=\Delta^*_J-a_J$.
	\begin{align}
	h(x_1)=\frac{d}{dx}P_{e|1}(x)\Big|_{M+\Delta^*_J} - \frac{1}{2}\frac{d}{dx}P_{e|1}(x)\Big|_{M}. \nonumber
	\end{align}
	Combining (\ref{d/dx P_e1}), (\ref{min_ai s}) and the above equation  results in $h(x_1)>0$. Similarly, we obtain $h(x_2)<0$.
	These inequalities, say that sign of the continuous function $h(x)$ at the sides of interval  $[0,\Delta^*_J-a_J]$ changes. According to the \emph{Bolzano's theorem}, $\exists\text{  } 0 <x_0 <\Delta^*_J-a_J \text{  }: h(x_0)=0$,
	the above equation means
	\begin{align}
	\frac{d}{dx}P_{e|1}(x)\Big|_{M+\Delta^*_J-x_0} = \frac{1}{2}\frac{d}{dx}P_{e|1}(x)\Big|_{M+x_0}. \nonumber
	\end{align}
	Using the above equation, we define 
	$	b_1=x_0 \text{ }, b_2=\Delta^*_J-x_0. \nonumber	$


\section{Convexity of $P_e$ in the presence of ISI} 
\label{proof_H_positive semi definite}

To show that  $P_e$ is a convex function of its arguments, we obtain the Hessian of $P_e$ and check whether or not it is positive semidefinite (PSD). Considering the definition of $s_t$, as the number of ``1"s after the last bit ``0", we obtain
\begin{align}
\mathbf{H}(P_e)&=\sum_{t=1}^{K}\frac{1}{2^{t+1}}\Big(\mathbf{H}(P_{e|s_t,0})+\mathbf{H}(P_{e|s_t,1})\Big), \label{hessian_sum}
\end{align}
where $\mathbf{H}(.)$ is the Hessian of the function. In the following, we show that both $\mathbf{H}(P_{e|s_t,0})$ and $\mathbf{H}(P_{e|s_t,1})$ are PSD. Therefore, $\mathbf{H}(P_e)$ is PSD too.
To do that, we have $\mathbf{H}(P_{e|B})=\Big[\frac{d^2}{d\Delta_id\Delta_j}P_{e|B}\Big]_{i,j=1, \cdots, K+1}$, $\text{ for }B=0,1$.
If the number of positive $\Delta_i$s equal to $K$, we have
	\begin{align}
	&\frac{d^2}{d\Delta_id\Delta_j}P_{e}=\sum_{t=1}^{K}\frac{1}{2^t}\frac{d^2}{d\Delta_id\Delta_j}P_{e|s_t} \nonumber\\
	&=\sum_{t=1}^{K}\frac{1}{2^t}\frac{1}{2}\Big(\frac{d^2}{d\Delta_id\Delta_j}P_{e|s_t,0}+\frac{d^2}{d\Delta_id\Delta_j}P_{e|s_t,1}\Big), \nonumber
	\end{align}
	where from (\ref{weigthed_sum}) we have
	\begin{align}
	&\frac{d^2}{d\Delta_id\Delta_j}P_{e|s_t,0}=\nonumber\\
	&c_ic_je^{-(c_0M+\sum_{i=1}^{K}c_i\Delta_i)}\frac{(T_\text{th}-(c_0M+\sum_{i=1}^{K}c_i\Delta_i)}{T_\text{th}!}. \nonumber
	\end{align}
	In signal $w_{0}$, we have $c_0M+...+c_K\Delta_K<T_\text{th}$ and $c_i,c_j\geq 0$. So we can write
	\begin{align}
	\frac{d^2}{d\Delta_id\Delta_j}P_{e|s_t,0}\Big|_{i,j=1}^{K}=c_ic_jF_{0}^+(w_{0}) \geq 0, \nonumber
	\end{align}
	where $F_{0}^+(w_{0})$ is a positive function.
	The Hessian matrix of $P_{e|s,0}$ is as follows.
	\begin{align}
	\mathbf{H}(P_{e|s,0})=
	F_{0}^+(w_{0})\left[
	\begin{array}{ccccc}
	{c_1}^2& c_1c_2  &\ldots& c_1c_K \\
	\vdots&\vdots&\ddots&\vdots \\
	c_1c_K & c_2c_K&\ldots &{c_K}^2
	\end{array} \right] \nonumber
	\end{align}
	and 
	$\underline{X}^T \mathbf{H}(P_{e|s,0})\underline{X}=F_{0}^+(w_{0})\Big(c_1X_1+...+c_K X_K\Big)^2\geq 0$.
	As a result, $\mathbf{H}(P_{e|s,0})$ is a PSD matrix. Similarly, we can show that $\mathbf{H}(P_{e|s,1})$ is also PSD.
	From (\ref{hessian_sum}), $\mathbf{H}(P_e)$ is a sum of PSD matrices. Thus, it is also a PSD matrix. Therefore, $P_e$ is a convex function of $\{\Delta_i\}_{i=1}^{K}$s.


\section{Adaptive threshold receiver} \label{Adaptive_th_no_isi}

In no-ISI case, we want the receiver to adjust its threshold such that $P_e$ decreases. To do so, we need to know the state of the system and the number of transmitted molecules in that state. Since $P_e$ is low, we ignore error propagation effects. Thus, the receiver uses its estimation instead of the system state and faces a hypothesis testing problem in each state. Using ML decoding rule in state $s_j$, results in
$
\frac{P(y|s_j,1)}{P(y|s_j,0)}  \underset{B_{j}=0}{\overset{B_{j}=1}\gtrless }1. \nonumber
$
Combining this  inequality and (\ref{poiss_dist}), we obtain the optimal threshold as
\begin{align}
T_\text{th}^{j}=\frac{M+\Delta^*_{j}}{\text{ln}\Big(1+\frac{M+\Delta^*_{j}}{\lambda}\Big)}. \nonumber
\end{align}
Note that the number of positive $\Delta^*_{i}$s is limited to $J$. So the number of different thresholds is also limited and for states $s_i \text{,}i=J+1,\cdots$ the threshold is fixed. As a result, the receiver needs a limited memory to save these thresholds. From (\ref{Pe0_Pe1_dTi}) error probability of adaptive threshold receiver is
\begin{align}
\tilde{P}_e= \sum_{y\geq T_\text{th}^j}^{} \frac{e^{-\lambda}}{2}\frac{{\lambda}^y}{y!} +\sum_{j=1}^{\infty} \sum_{y < T_\text{th}^j}^{}\frac{e^{-(M+\Delta^*_j+\lambda)}}{2^{j}}\frac{{(M+\Delta^*_j+\lambda)}^y}{y!}.  
\label{pe_adaptive_adaptive Th_1iter}
\end{align}

\begin{remark} \label{remark:min Pe_opt delta & opt th _ simmultaneously}
	We can optimize the release duration increments and thresholds simultaneously in each state by solving the following optimization problem.
\begin{equation}
\min_{\begin{array}{cc}
	{\{\Delta_i\}_{i=1}^{\infty}} \\
	{\{T_\text{th}^i\}_{i=1}^{\infty}}
	\end{array}}{\tilde{P}_e}, \text{ }\text{ }\text{ s. t. :} \left\{
\begin{array}{rl}
0 \leq \sum_{j=1}^{i}\Delta_j \leq B_M \text{ },\forall i, \\
0<T_\text{th}^i<M+\Delta_i\text{ },\forall i,
\end{array} \right.  \nonumber
\end{equation}
 \end{remark}


\section{Sub-optimal solution for one-symbol ISI} \label{appx_sol_1bit_ISI}

A local minimum of (\ref{min_pe_isi_b}) can be found if $C_{1i}$ is  active for $i = J,J+1,\cdots $ ($J$ is the number mentioned in \emph{Lemma} \ref{lem_2}),  and $\mu^*_{1i}=\mu^*_{2i} \geq 0$ for $i =J+1,\cdots$. Therefore, $C_{2i}$ is  active for $i =J+1,\cdots $, \emph{i.e.,} $l_i=M$ and 
\begin{align}
v_i=p_1M=v_{J+1}, \text{ for } i =J+1,\cdots . \label{fixed_isi after J}
\end{align}
Thus, $\mu^*_{1i}= 0$ for $i=1,\cdots,J-1$ and $\mu^*_{2i}= 0$ for $i=1,\cdots,J$.
If we define $U_i=[u_i(1),u_i(2),\cdots]^T$, where
\begin{subnumcases}{u_i(j)=}
-1, j \leq i,   \nonumber \\  
0\text{ }\text{ }, j>i, \nonumber
\end{subnumcases}
then, (\ref{kkt_isi_b}) results in
$\nabla\mathbb{E}_I[P_{e|1}(p_0l_i^*+v_i^*)] =\mu^*_{1J}U_J$.
Combining (\ref{kkt_isi_b}) and the above equation results in
\begin{align}
\left(\begin{array}{ccc} \frac{1}{2} \frac{d}{dl_1}P_{e|1}(p_0l_1+v_1)\\
\vdots\\
\frac{1}{2^J}\frac{d}{dl_J}P_{e|1}(p_0l_J+v_J)
\end{array} \right)=
\frac{\mu_{1J}^*}{p_0}\left(\begin{array}{ccc} -1 \\ \vdots\\-1\end{array} \right). \label{kkt_isi2}
\end{align}\\
Form (\ref{fixed_isi after J}),  we have 
$\text{P}(I=v_{J+1})=\sum_{j=J+1}^{\infty}\text{P}(s_{i-1})=\frac{1}{2^{J}}. \nonumber $
As a result, ISI takes $J+1$ different values.
Comparing (\ref{kkt_isi2}) with (\ref{g2}), we conclude that these two problem have the same solution, which is the optimal solution of no-ISI case. If we denote the optimal released molecule increments in no-ISI case by $\Delta_{i, \text{no-ISI}}^{*}$, the solution of (\ref{kkt_isi2}) will be
\begin{align}
p_0l_i^*+v_i^*=M+\Delta_{i, \text{no-ISI}}^{*} \text{ for }i=1,\cdots,J. \nonumber
\end{align}
Recall that $v_1$ is the ISI value in state $s_0$ (the previous bit is ``0"). Thus,  $v_1^*=0$ and from the above equation, $l_1^*=\frac{1}{p_0}(M+\Delta_{1, \text{no-ISI}}^{*})$. For $i \geq J+1$, we have $l_i=M$ and for $2\leq i \leq J$, we have  $v_i^*=p_1 l^*_{i-1}$ and $l_i^*$ is recursively obtained as
\begin{align}
l_i^*=\frac{1}{p_0}( M+\Delta_{i, \text{no-ISI}}^{*}- p_1 l^*_{i-1}), \text{ for }i=2,\cdots,J. \label{li*s_appx}
\end{align}


\section{Sub-optimal solution for two-symbol ISI} \label{appx_sol_2bit_ISI}

If we assume that $v_{2i-1}$s for $i=1,2,\cdots$ are equal ($v_{2i-1}=v_1$ for $i=1,2,\cdots$), and substitute their average value in non-adaptive strategy instead of $v_1$, we can simplify (\ref{pe_isi_2b}) and obtain a sub-optimal solution ($l^*_i$s).
\begin{align}
v_1&=\mathbb{E}[\{v_{2i-1}\}_{i=1}^\infty]=\frac{\sum_{i=1}^{\infty}v_{2i-1}P_{\tilde{s}_{2i-1}}}{\sum_{i=1}^{\infty}P_{\tilde{s}_{2i-1}}} \nonumber
\\
&=\frac{\frac{1}{2}(0)+p_2M(\frac{1}{4}+\frac{1}{8}+\cdots)}{\frac{1}{2}+\frac{1}{4}+\cdots}=\frac{p_2M}{2}.\nonumber
\end{align}
Substituting the above  equation and (\ref{state_dist_2ISI}) in   (\ref{pe_isi_2b}) results in the following minimization problem:
\begin{align}
\min_{\{l_1\}_{i=1}^\infty} \mathbb{E}_I[P_{e|1,I}]=&\frac{1}{2}P_{e|1}(p_0l_1+v_{2i-1}) \nonumber\\
&+\sum_{i=1}^{\infty} \frac{1}{2^{i+1}} P_{e|1}(p_0l_{2i}+v_{2i}),
\nonumber
\end{align}
subject to (\ref{C1i_isi_b}) and (\ref{C2i_isi_b}).
Using Lagrangian method, results in
\begin{align}
\nabla\mathbb{E}_I[P_{e|1,I}],=\sum_{k=1}^{2}\sum_{i=1}^{\infty}\mu^*_{ki}\nabla C_{ki}(\{l_i^*\}_{i=1}^{\infty}),  \label{kkt_2isi_bbbb}
\end{align}
where $\mu^*_{ki}$s are Lagrangian multipliers. Due to KKT conditions, we have $\mu^*_{ki} C_{ki}=0, \forall i$. 

A local minimum of (\ref{min_pe_isi_b}) is obtained if $C_{1i}$ is  active for $i \geq J $ ($J$ is the number mentioned in \emph{Lemma} \ref{lem_2}),  and $\mu^*_{1i}=\mu^*_{2i} \geq 0$ for $i \geq J+1$. Therefore, $C_{2i}$ is  active for $i \geq J+1$ (\emph{i.e.}, $l_i=M$ and 
$v_i=p_1M, \text{ for } i \geq J+1 $).
So $\mu^*_{1i}= 0$ for $i=1,\cdots,J-1$ and $\mu^*_{2i}= 0$ for $i=1,\cdots,J$.
Thus, (\ref{kkt_2isi_bbbb}) results in
\begin{align}
\left(\begin{array}{ccc} \frac{1}{2}\frac{d}{dl_1}P_{e|1}(p_0l_1+v_1)
\\ \vdots 
\\ \frac{1}{2^J}\frac{d}{dl_J}P_{e|1}(p_0l_J+v_J)
\end{array} \right)  =\frac{\mu^*_{1J}}{p_0}\left(\begin{array}{ccc} -1 \\ \vdots\\-1\end{array} \right). \label{grad_2bit_isi}
\end{align}
Comparing (\ref{grad_2bit_isi}) and (\ref{g2}), we conclude that these equations have the same solution which is the optimal solution of no-ISI case. Therefore, the solution of (\ref{grad_2bit_isi}) will be
\begin{align}
p_0l_i^*+v_i^*=M+\Delta_{i, \text{no-ISI}}^{*} \text{ for }i=1,\cdots,J. \nonumber
\end{align}
Recall that $v_1$ is the ISI value in state $s_1$ (the previous bits are ``0"). Thus,  $v_1^*=0$ and from the above equation, $l_1^*=\frac{1}{p_0}(M+\Delta_{1, \text{no-ISI}}^{*})$. For $i\geq2$, we have  $v_i^*=p_1 l^*_{i-1}+p_2l^*_{i-2}$ and $l_i^*$ is recursively obtained as
\begin{align}
l_i^*=\frac{1}{p_0}(M+\Delta_{i, \text{no-ISI}}^{*}-{p_1} l^*_{i-1}-p_2 l^*_{i-2}), 2\leq i \leq J. \label{li*s_appx_2bit_isi}
\end{align}


\section{Closed form rates of one-symbol ISI in \cite{movahednasab2016adaptive}} \label{movahed_1mem}

In \cite{movahednasab2016adaptive}, if the transmitter starts with a ``0" bit, $X_1=0$ (state $s_0$) and 
$X_2=B_2M $ (state $s_{B_2}$).
For $i \geq 3$ if $B_i=0$, the ISI resets (state $s_0$) and if $B_i=1$ (state $s_1$), we have
$X_i=M-\frac{p_1}{p_0} X_{i-1}. \nonumber$
We simplify $X_{i}$ in state $s_j$ recursively~as
\begin{align}
X_{i,s_j}=B_iM \sum_{k=0}^{j}(-\frac{p_1}{p_0})^{k}. \nonumber
\end{align}
The conditional probability of error is
$P_{e|1,s_j}=P_{e|1}({p_0}M), \forall j\in \mathbb{N},\nonumber$
and we have
$P_{e|0,s_j}=P_{e|0}(p_1 X_{i,s_j}).$
Using the probability of states in (\ref{state_probes}) causes
\begin{align}
P_{e|0}&=\sum_{j=0}^{\infty}(\frac{1}{2})^{j+1}\Big(P_{e|0}(p_1 X_{i,s_j})\Big). \nonumber
\end{align}
$P_{e|1}$ is a constant function but $P_{e|0}$ is an increasing function of its argument. Because
$\frac{d}{dx}P_{e|0}(x)=e^{-(x+\lambda)}\frac{(x+\lambda)^{T_\text{th}}}{T_\text{th}!}>0. \nonumber$
We use this result to find upper and lower bounds on $P_{e|0}$.
\\
$\text{I.}$ \emph{Upper bound}:
	if $B_i=1$, we have
	\begin{align}
	X_{i,s_j}&<M\Big(1-(\frac{p_1}{p_0})+\cdots+(\frac{p_1}{p_0})^{2k}\Big)  \nonumber\\
	&:=\tilde{X}_{2k} ,\text{ }0<2k<j, k\in \mathbb{N} \nonumber
	\end{align}
	According to the above inequality, other terms ($j>2k$) are less than an even term. Therefore, we have
	\begin{align}
	P_{e}=&\frac{1}{2}\Big(\frac{1}{2}(P_{e|00}+P_{e|10})+P_{e|1}\Big)\nonumber\\
	=&\frac{1}{2}\Big(\frac{1}{2}\big(P_{e|0}(0)+\sum_{j=1}\frac{1}{2^j}P_{e|0}(p_1X_{i,s_j})\big)\Big)+\frac{1}{2}P_{e|1} \label{apx_pe_zoj} \\
	<&\frac{1}{4}P_{e|0}(0)+\sum_{j=1}^{2k}\frac{1}{2^{j+2}}P_{e|0}(p_1X_{i,s_j})\nonumber\\
	&+\frac{1}{2^{2k+2}}P_{e|0}(p_1\tilde{X}_{2k})+\frac{1}{2}P_{e|1} :=\tilde{P}_{e,2k}. \label{movahed_upper B_1mem}
	\end{align}\\
$\text{II.}$ \emph{Lower bound}:
	if $B_i=1$, we have
	\begin{align}
	X_{i,s_j}&>{M}\Big(1-(\frac{p_1}{p_0})+...-(\frac{p_1}{p_0})^{2k-1}\Big)  \nonumber\\
	&:=\tilde{X}_{2k-1}   ,\text{ }0<2k<j, k\in \mathbb{N}. \nonumber
	\end{align}
	According to the above inequality, other terms ($j>2k-1$) are more than an odd term. Therefore, using (\ref{apx_pe_zoj}) we have 
	\begin{align}
	P_{e}>&\frac{1}{4}P_{e|0}(0)+\sum_{j=1}^{2k-1}\frac{1}{2^{j+2}}P_{e|0}(p_1X_{i,s_j})\nonumber\\
	&+\frac{1}{2^{2k+1}}P_{e|0}(p_1\tilde{X}_{2k-1})+\frac{1}{2}P_{e|1}
	:=\tilde{P}_{e,2k-1}. \label{movahed_lower B_1mem}
	\end{align}

\end{appendices}


\bibliographystyle{ieeetr}
\bibliography{ref}
 

\end{document}